\documentclass[10pt, oneside]{article}

\usepackage[utf8]{inputenc}

\usepackage{amsmath, amsthm, amssymb, wasysym, verbatim, bbm, color, graphics, geometry, tikz, hyperref, cleveref, braket, tikz-cd}

\usepackage{dsfont}

\usepackage{multirow}
\usepackage{tablefootnote}

\newtheorem{theorem}{Theorem}[section]
\newtheorem{definition}[theorem]{Definition}

\newtheorem{remark}[theorem]{Remark}
\newtheorem{lemma}[theorem]{Lemma}
\newtheorem{cor}[theorem]{Corollary}
\newtheorem{prop}[theorem]{Proposition}
\newtheorem{assumption}[theorem]{Assumption}

\newtheoremstyle{named}{}{}{\itshape}{}{\bfseries}{.}{.5em}{\thmnote{#3}}
\theoremstyle{named}
\newtheorem*{namedtheorem}{Theorem}

\newcommand{\mcP}{\mathcal{P}}
\newcommand{\mcC}{\mathcal{C}}

\newcommand{\mcA}{\mathcal{A}}
\newcommand{\mcB}{\mathcal{B}}
\newcommand{\mcO}{\mathcal{O}}
\newcommand{\mcQ}{\mathcal{Q}}

\newcommand{\mcD}{\mathcal{D}}

\newcommand{\mcN}{\mathcal{N}}

\DeclareMathOperator{\id}{\mathds{1}}

\DeclareMathOperator{\Tr}{Tr}
\DeclareMathOperator{\poly}{poly}
\DeclareMathOperator{\Prob}{Pr}
\DeclareMathOperator{\polylog}{polylog}

\usepackage{float}

\usepackage[sorting=nyt, backend = bibtex, style=alphabetic, backref]{biblatex}
\DefineBibliographyStrings{english}{%
  backrefpage = {p.},
  backrefpages = {pp.},
}
\AtNextBibliography{\small}

\usetikzlibrary{decorations.pathmorphing}

\tikzset{snake it/.style={decorate, decoration=snake}}

\newcommand{\ptp}{\mathsf{\Gamma}}

\newcommand{\w}{0.23cm}

\newcommand{\bx}[3]{
  \draw[fill = white] #3 (#1*\w-\w/2,-#2*\w-\w/2) rectangle (#1*\w+\w/2,-#2*\w+\w/2);
}

\newcommand{\yd}[2][0.4]{
  \begin{tikzpicture}[scale = #1, baseline={([yshift=-0.6ex]current bounding box.center)}]
    \foreach \li [count = \y] in {#2} {
      \foreach \x in {1,...,\li} {
        \bx{\x}{\y}{}
      }
    }
  \end{tikzpicture}
}

\title{The mixed Schur transform: efficient quantum circuit and applications}
\author{Quynh T. Nguyen\thanks{Harvard University. Email: qnguyen@g.harvard.edu}}
\date{}

\begin{document}

\maketitle

\begin{abstract}The Schur transform, which block-diagonalizes the tensor representation $U^{\otimes n}$ of the unitary group $\mathbf{U}_d$ on $n$ qudits, is an important primitive in quantum information and theoretical physics. We give a generalization of its quantum circuit implementation due to Bacon, Chuang, and Harrow~\cite{bacon2007quantum} to the case of mixed tensor $U^{\otimes n} \otimes \Bar{U}^{\otimes m}$, where $\Bar{U}$ is the dual representation. This representation is the symmetry of unitary-equivariant channels, which find various applications in quantum majority vote, multiport-based teleportation, asymmetric state cloning, black-box unitary transformations, etc.
The ``mixed'' Schur transform contains several natural extensions of the representation theory used in the Schur transform, in which the main ingredient is a duality between the mixed tensor representations and the walled Brauer algebra. Another element is an efficient implementation of a ``dual'' Clebsch-Gordan transform for $\Bar{U}$. The overall circuit has complexity $\widetilde{\mcO}((n+m)d^4)$\footnote{The $\widetilde{\mcO}(\cdot)$ notation hides polylog factors in $n$, $m$, $d$, and $1/\varepsilon$, where $\varepsilon$ is the target precision in diamond norm. 
 }. Finally, we show how the mixed Schur transform enables efficient implementation of unitary-equivariant channels in various settings and discuss other potential applications, including an extension of permutational quantum computing that includes partial transposes.

\end{abstract}

\section{Introduction}
Representation theory plays an important role in quantum information~\cite{hayashi2017group}. For example, at the heart of Shor's algorithm is the quantum Fourier transform, which is the decomposition of the regular representation of the cyclic group into irreducible representations (irreps). Hidden subgroup problems are generalizations of period finding problems to nonabelian groups~\cite{childs2010quantum} and thus nonabelian group quantum Fourier transforms are also of great interest \cite{moore2006generic}. However, irrep decompositions of representations other than regular ones may also be valuable. On the one hand, they may serve as subroutines in quantum information and communication protocols. On the other hand, they go beyond the standard paradigm of quantum Fourier transforms, and thus may provide new directions for algorithmic quantum speedups~\cite{bravyi2023quantum, zheng2022super}. The most important decomposition of this kind is perhaps the quantum Schur transform~\cite{harrow2005applications, bacon2007quantum}. From the perspective of the unitary group $\mathbf{U}_d$, it block-diagonalizes the tensor representation, $\phi^d_n(U) = U^{\otimes n}$, into irreps

\begin{equation}
    U^{\otimes n} \overset{\text{Schur}}{\longrightarrow} \bigoplus_{\lambda \text{ irreps}} \mathbf{q}^d_{\lambda}(U) \otimes \id_{m_\lambda}, \forall U \in \mathbf{U}_d,
    \label{eq:schur-weyl}
\end{equation}
where $m_\lambda$ is the multiplicity of the irrep $\mathbf{q}^d_{\lambda}$ (see definitions later). It has many applications in quantum information, such as state tomography~\cite{haah2017sample}, hypothesis testing~\cite{hayashi2002optimal}, decoherence-free subsystem encoding~\cite{kempe2001theory}, communication without shared reference frames~\cite{bartlett2003classical}, and quantum machine learning \cite{zheng2022super}. The quantum Schur transform and these applications take advantage of the celebrated Schur-Weyl duality, which describes a beautiful interplay between the unitary group and the symmetric group (see Section~\ref{sec:SW}).

In this paper, we are interested in a generalization of the tensor representation and Schur-Weyl duality. In particular, we consider the \textit{mixed tensor representation}
\begin{equation}
    \phi_{n,m}^d (U) = U^{\otimes n} \otimes \Bar{U}^{\otimes m}.
\end{equation}

On the mathematical side, our motivation is the existence of the \emph{mixed Schur-Weyl duality} and a large body of mathematics literature studying extensions of the tensor representation and the symmetric group into this mixed tensor representation and the so-called \emph{walled Brauer algebra}~\cite{benkart1994tensor}. From the perspective of $\mathbf{U}_d$, the mixed Schur transform performs the following decomposition
\begin{equation}
    U^{\otimes n} \otimes \Bar{U}^{\otimes m} \overset{\text{mixed}  \atop \text{Schur}}{\longrightarrow} \bigoplus_{\gamma \text{ irreps}} \mathbf{q}^d_\gamma (U) \otimes \id_{m_\gamma}, \forall U \in \mathbf{U}_d,
    \label{eq:mixedSchur}
\end{equation}
where $\gamma$ simultaneously label the \textit{rational} irreps of $\mathbf{U}_d$ and of the walled Brauer algebra (definitions given later). In condensed matter physics, mixed tensor representation appears in the study of integrable antiferromagnetic $gl(M|N)$ superspin chain of alternating fundamental and dual representations~\cite{candu2011continuum}. Similar kinds of symmetries also appear in various areas of high energy physics~\cite{kimura2010free}, including gauge-gravity duality~\cite{kimura2007branes, kimura2008enhanced}. Generalizing the quantum Schur transform into this case is thus of great mathematical interest.

On the quantum information side, the mixed tensor representation appears most predominantly in unitary-equivariant channels\footnote{The term ``covariance'' is also used. To our knowledge, ``equivariance'' is preferred in cases where the input and output representations are ``similar,'' such as $U^{\otimes n}$ and $U^{\otimes m}$ in this work.} An $m$ to $n$ qudit channel is said to be unitary-equivariant if it satisfies the following property
\begin{equation}
    \mathcal{N}(U^{\otimes m} \rho U^{\dagger \otimes m}) = U^{\otimes n}  \mathcal{N}(\rho )U^{\dagger \otimes n}, \qquad \forall U \in \mathbf{U}_d.
\end{equation}
Unitary-equivariant channels~\cite{grinko2022linear} have a wide variety of applications in quantum information: quantum majority vote~\cite{buhrman2022quantum}, asymmetric cloning~\cite{nechita2021geometrical}, port-based teleportation~\cite{ishizaka2008asymptotic, kopszak2021multiport}, covariant quantum error correction~\cite{kong2022near}, equivariant variational algorithms~\cite{nguyen2022theory}, entanglement detection~\cite{huber2022dimension}, black-box transformations~\cite{yoshida2022reversing}, etc.
In this case, it can be shown that the Choi-Jamiolkowski state of the channel, $J^{\mathcal{N}} \triangleq  \frac{1}{d^m}\sum_{i,j=1}^{d^m} \ket{i}\bra{j} \otimes \mcN (\ket{i}\bra{j})$, commutes with the mixed tensor representation~\cite{wolf2012quantum}
\begin{equation}
    [J^\mathcal{N}, \Bar{U}^{\otimes m} \otimes U^{\otimes n}]=0.
\end{equation}
Despite their wide applications, efficient implementations of these channels, to our knowledge, have been rather underexplored.
Our main motivation is that the mixed Schur transform may enable efficient optimization and implementation of these channels. Beyond applications in unitary-equivariant channels, the mixed Schur transform may provide new possibilities for generalizing existing quantum state or channel learning and property testing protocols and developing new quantum algorithms.

\subsection{Results}
We construct an efficient quantum circuit for the mixed Schur transform that effects the block-diagonalization in Eq.~\eqref{eq:mixedSchur}.

\begin{namedtheorem}[Main result] \emph{(Corollary~\ref{cor:main})} Consider a system of $n+m$ qudits and fix a standard computational basis $\ket{i}$, $i \in [d]$ for each qudit. Then the block-diagonalization in Eq.~\eqref{eq:mixedSchur} can be effected with an isometry with gate complexity $\mcO((n+m)d^4 \polylog (n,m,d,1/\varepsilon))$, where $\varepsilon$ is the desired precision $\varepsilon$ in diamond norm. After the mixed Schur transform, the basis is relabelled as $\ket{\gamma}\ket{q_\gamma}\ket{p_\gamma}$, where the registers encodes the irrep label, irrep space basis element, and multiplicity space basis element, respectively, according to a canonical choice of basis (see Section~\ref{sec:mixedSW} for detailed definitions), such that $U^{\otimes n} \otimes \Bar{U}^{\otimes m}$ is block-diagonal in this basis.
\end{namedtheorem}

The construction of the mixed Schur transform is based on two main ingredients:

First, we develop a new scheme of labeling the (mixed) Schur basis based on the branching rules of the $\mathbf{U}_d$ \emph{rational} irreps and the walled Brauer algebra. Specifically, rational irreps show up because the matrix entries $\Bar{U}=(U^{-1})^\top$  are rational functions in the entries of $U$. This is in contrast to~\cite{bacon2007quantum}, where the authors only concerned with polynomial irrep that can be labeled by Young diagrams. Here, we label the rational irrep registers $\ket{\gamma}$ and $\ket{q_\gamma}$ by combinatorial objects called staircases (Section~\ref{sec:staircase})~\cite{stembridge1987rational} and Gelfand patterns (Section~\ref{sec:GTbasis}) (the latter is equivalent to rational tableaux~\cite{stembridge1987rational,kwon2008rational}), whose combinatorics generalizes and is somewhat more involved than that of polynomial irreps.  On the other hand, the labeling of the multiplicity space $\ket{p_\gamma}$ is obtained by invoking the mixed Schur-Weyl duality~\cite{koike1989decomposition, benkart1994tensor}, which relates the centralizer algebra of $U^{\otimes n} \otimes \Bar{U}^{\otimes m}$ to a representation of the walled Brauer algebra, a generalization of the symmetric group algebra. Remarkably, the branching rules in both of the Schur and mixed Schur transforms are multiplicity-free. This allows us to construct a ``subalgebra-adapted'' basis for $\ket{p_\gamma}$ (Section~\ref{sec:bratteli}).

\begin{table}[h!]
\centering
\begin{tabular}{| c| c| c |} 
  \hline
  & \textbf{Schur}~\cite{bacon2007quantum} & \textbf{Mixed Schur} \\ 
  \hline
  Representation & $U^{\otimes n}$ & $U^{\otimes n} \otimes \Bar{U}^{\otimes m} $  \\
  \hline
  Irrep label & Young diagrams: $\yd[1]{4,3,1}$  & Staircases: $\left[\yd[1]{3,1} , \yd[1]{2} \right]$ \\
  \hline
  Irrep space label & Gelfant-Tsetlin $\equiv$ SSYT & GT with negative entries  \\
  \hline
  Multiplicity space label & Young-Yamanouchi $\equiv$ SYT & Up-down staircase tableaux~\cite{stembridge1987rational} \\
  \hline
  Centralizer & \begin{tabular}{c}
      Schur-Weyl duality:  \\
        Symmetric group algebra
  \end{tabular} & \begin{tabular}{c}
      Mixed Schur-Weyl duality: \\
       Walled Brauer algebra 
  \end{tabular}   \\
  \hline
  Branching rules & Multiplicity-free &  Multiplicity-free\\
  \hline
  Primitives & Clebsch-Gordan transform & CG and dual CG transforms \\
  \hline 
  Circuit complexity &  $\widetilde{\mcO}(nd^3)$ & $\widetilde{\mcO}((n+m)d^4)$ \\
  \hline
\end{tabular}
\caption{Comparison between the Schur and mixed Schur transforms.}
\end{table}

Second, we construct a ``dual'' CG transform for input $\Bar{U}$ (Section~\ref{sec:recursive}), which performs the irrep decomposition of the tensor product of an arbitrary \emph{rational} irrep and the dual irrep $\bar{U}$. This generalizes a previous construction in~\cite{buhrman2022quantum} who solved the task for $d=2$. Similar to~\cite{bacon2007quantum}, we use the theory of $\mathbf{U}_d$ tensor operators~\cite[Chapter 18]{vilenkin1992representations} to develop a recursive circuit of the dual CG transform. We obtain a circuit complexity of $\widetilde{\mcO}(d^4)$. In fact, we further show that extensions to more general CG transforms of so-called elementary irreps can be achieved as well.

\paragraph{Applications} We show that the mixed Schur transform enables efficient implementation of unitary-equivariant channels in various settings, partially answering an open question in~\cite{grinko2022linear} who asked how to implement the channel knowing the Choi state in the (mixed) Schur basis.
We give three different implementations in Section~\ref{sec:implementation}
\begin{itemize}
    \item A teleportation-based method: It assumes the Choi state in the Schur basis can be prepared as a quantum state, then performs a modified teleportation protocol that exploits the unitary equivariance property of the channel. The channel is implemented exactly, \textit{deterministically} when $m=1$ (number of input qudits) and probabilistically with constant success probability when $m=O(1)$. This is the case for applications in state cloning~\cite{fan2014quantum, nechita2021geometrical}, covariant error correction~\cite{kong2022near}, equivariant variational algorithms~\cite{nguyen2022theory} (cf. Section~\ref{sec:example}).  The main advantage of this method is that it outputs the exact full output state.
    \item A shadow tomography-based method: It assumes the Choi state in the Schur basis can be prepared as a quantum state and allows estimating Pauli-sparse observable expectation values on the output. Examples of such observables include any Pauli operator or any local observables. It is based on the protocols in Ref.~\cite{caro2022learning, huang2022learning} and gives either average-case performance guarantee, or worst-case guarantee under the assumption the input is also Pauli-sparse. The main advantage of this method is it requires very few quantum processing resources (only Bell measurements or single-qubit measurements are needed). When $n$ (number of output qudits) is small, the full output state can be efficiently reconstructed, but this is not possible in general.
    \item A block-encoding-based method: It assumes the Schur-basis entries of the Choi state can be computed efficiently and allows estimating the inner product of the output with another state or the expectation value of low-rank observables. The main advantage is that it covers observables that are not Pauli-sparse and that composing multiple channels is straightforward. Its main disadvantage is assumptions in the spectral properties of the Choi state which necessarily require the channel to have high Kraus rank.
\end{itemize}
In concurrent works, Ref.~\cite{grinko2023gelfand} have used the mixed Schur transform to efficiently implement the port-based teleportation protocol~\cite{ishizaka2008asymptotic}. Ref.~\cite{fei2023efficient} also achieved the same result using what they call the twisted Schur transform. We thank the authors for informing us about these results.

In Section~\ref{sec:other-applications}, we define an extension of the complexity class permutational quantum computing (PQP) introduced by Jordan~\cite{jordan2010permutational} by combining with Hamiltonian simulation ideas from~\cite{zheng2022super}.
We call it ``partially transposed'' PQP and conjecture it is not classically simulable.
In particular, the model is defined by starting from a computational basis state, performing Hamiltonian evolution of $H$ which is a linear combination of walled Brauer diagrams, and then sampling in the mixed Schur basis. This allows estimating the matrix elements $|\bra{\gamma,q_\gamma,p_\gamma}_\mathrm{Sch} e^{-iHt} \ket{\mu,q_\mu,p_\mu}_\mathrm{Sch}|$ to additive error, a task generalizing strong Schur sampling~\cite{havlivcek2018quantum} to the mixed Schur case.
Physically, this model includes Hamiltonian evolutions of antiferromagnetic $gl(M|N)$ superspin chains such as those studied in~\cite{candu2011continuum}. We also discuss other potential applications of the mixed Schur transform in learning and symmetry testing, quantum SDP solvers with symmetry constraints, as well as potential applications to quantum complexity and algorithms of representation-theoretic problems regarding the walled Brauer algebra~\cite{ikenmeyer2023remark, bravyi2023quantum}.

\subsection{Discussion}\label{sec:discussion}
Our work provides a new efficient transform into the quantum information toolbox and opens up possibilities for many interesting follow-up studies. We believe the mixed Schur transform can find more applications in quantum information theory and related fields. 

One interesting question is whether the polynomial dependence on $d$ can be improved to $\polylog(d)$.
In the Schur transform case, Ref.~\cite{krovi2019efficient} achieves this by a construction based on the $\mathfrak{S}_n$ quantum Fourier transform and so-called permutation modules. Does a similar construction exist by starting from the walled Brauer algebra? It is an open question how to define or implement a Fourier transform for the walled Brauer algebra and, if possible, how it would be related to the mixed Schur transform. However, we note that such a transform would also have implications to the quantum complexity of certain representation-theoretic problems discussed in Section~\ref{sec:other-applications}.

As a remark, the mixed Schur-Weyl duality (cf. Section~\ref{sec:mixedSW}) implies that the mixed Schur transform is also an irrep decomposition of a representation of the walled Brauer algebra. Therefore, to our knowledge, our work gives the first quantum transform of an algebra representation that is not simply a group algebra (those generated by group representations). Group-related transforms such as group quantum Fourier transforms \cite{beals1997quantum, moore2006generic} have received much attention since the beginning of quantum computation. What are the roles of algebra-related transforms (those are not just group algebras)? Would they be useful in some practical algorithmic tasks?

\paragraph{Paper organization}
In Section~\ref{sec:reptheory}, we provide the necessary representation theory background for the mixed Schur transform, including rational representations of the unitary group, the walled Brauer algebra, and the mixed Schur-Weyl duality, and develop an irrep labeling system based on them. Then, we construct a quantum circuit for the mixed Schur transform in Section~\ref{sec:transform} assuming efficient implementations of the Clebsch-Gordan and dual Clebsch-Gordan transforms. We give an efficient construction of the dual Clebsch-Gordan transform in Section~\ref{sec:recursive}. Finally, we discuss applications of the mixed Schur transform in Section~\ref{sec:applications}.

\paragraph{Acknowledgements}
We would like to thank Dmitry Grinko for valuable discussions and his pointer to Ref.~\cite{vilenkin1992representations}, thank the authors of Ref.~\cite{grinko2023gelfand} and Ref.~\cite{fei2023efficient} for coordinating to concurrently upload our papers, and thank Aram Harrow, Anurag Anshu, and Jonas Haferkamp for valuable feedback on drafts of this work. The author acknowledges support from the NSF Award No. 2238836 via Anurag Anshu and support from the Harvard Quantum Initiative.

\section{Preliminaries}\label{sec:reptheory}
In this section, we describe the generalization of the representation theory previously used in the Schur transform. The central tool here is the \textit{mixed} Schur-Weyl duality, which asserts a duality between the mixed tensor representation of $\mathbf{U}_d$ and the partially tranposed permutation representation of the walled Brauer algebra. We will then use the combinatorial tools first developed by~\cite{stembridge1987rational} to derive a basis labeling scheme for quantum implementation of the Schur transform.

\subsection{Representations of $\mathbf{U}_d$ and $\mathfrak{S}_n$}\label{sec:SW}
The representation theory of $\mathbf{U}_d$ and $\mathfrak{S}_n$ is ubiquitously studied in quantum information literature (see \cite[Sections 5-7]{harrow2005applications} for an overview), here we briefly recall the basics of representation theory used in the Schur transform~\cite{bacon2007quantum}.

Let $\mathrm{GL} (V)$ denote the group of all invertible linear transformations from a vector space $V$ to itself. A representation of a group $G$ on a complex vector space $V$ is a homomorphism $R: G \mapsto \mathrm{GL}(V)$. When $G$ is a matrix group, we call $R$ a \textit{polynomial} (resp. \textit{rational}) representation if, after choosing a basis for $V$, the matrix entries of the representation $R(g)$ of a matrix $g \in G$ are fixed polynomials (resp. rational functions) in the entries of $g$. For example, the map
\begin{equation}
        \begin{pmatrix}
       a & b \\
       c & d
    \end{pmatrix} \longrightarrow \begin{pmatrix}
       a^2 & 2ab & b^2 \\
       ac & ad + bc & bd\\
       c^2 & 2cd & d^2
    \end{pmatrix}
\end{equation}
is a polynomial representation of $G = \mathbf{U}_d$ (and also of $\mathbf{SU}_d, \mathbf{GL}_d, \mathbf{SL}_d$~\cite[Section 18]{vilenkin1992representations}) of dimension three, while
\begin{equation}
    \begin{pmatrix}
       a & b \\
       c & d
    \end{pmatrix} \longrightarrow (ad -bc)^{-1}
\end{equation}
is a rational representation of dimension one. The latter example is called the inverse determinant irrep (denoted $\mathrm{det}^{-1})$. More generally, it can be shown that the $r$-power \textit{determinant irreps}, denoted by $\mathrm{det}^r$, enumerate all one-dimensional representations of $\mathbf{U}_d$ (of course, they coincide for the groups $\mathbf{SU}_d$ and $\mathbf{SL}_d$ whose elements have unit determinant).

 If $G$ is a compact group, then any representation can be decomposed into irreducible representations (irreps). In other words, under the action of $G$ via the representation $R$, we can view $V$ as a direct sum of $V= V_1 \oplus V_2 \oplus \hdots$, where each $V_i$ is nonzero and invariant under the $G$'s action, and no $V_i$ has a proper $R(G)$-invariant subspace. The multiset of irreps is unique up to equivalence (i.e. under invertible transformations). Thus, to determine a representation $R$ up to equivalence, it suffices to describe the multiplicity of each irrep appearing in the decomposition of $R$. In this paper, we consider $\mathbf{U}_d$, the group of $d$-dimensional unitary transforms. The mixed Schur transform we will construct works the same way for other related groups such as $\mathbf{SU}_d, \mathbf{GL}_d, \mathbf{SL}_d$ because they essentially share the same irrep structures.

Much of quantum information literature studies the symmetric group $\mathfrak{S}_n$ and the tensor representation $U^{\otimes n}$, which is a polynomial representation of $\mathbf{U}_d$. The polynomial irreps of $\mathbf{U}_d$ and of $\mathfrak{S}_n$ can be labeled by \textit{Young diagrams}. A Young diagram is a sequence of weakly decreasing nonnegative integers $\lambda = (\lambda_1, \lambda_2, \hdots )$. It is well-known that all the polynomial irreps of $ \mathbf{U}_d$  are in one-to-one correspondence (up to equivalence) with $\lambda$ of length at most $d$. We visualize a Young diagram by a left-justified diagram of boxes, where the first row has $\lambda_1$ boxes, and so on. If some $\lambda_i=0$ we oftentimes omit them in the notations.
For examples, the \textit{defining representation}, defined by $\mathbf{q}^d_{\lambda}(U)= U, \forall U \in \mathbf{U}_d$, is labeled by $\lambda =(1) \equiv \square$; the determinant irrep, $\mathrm{det}(U)$, is labeled by $(1^d)$ (a column of $d$ boxes); and the irrep $\lambda=(4,2,1)$ is
\begin{equation*}
  \yd[1.5]{4,2,1}.
\end{equation*}

Furthermore, it can be shown that the tensor representation $\phi_n^d(U)=U^{\otimes n}$ can be decomposed into polynomial irreps labeled by Young diagrams $\lambda$ of length at most $d$ that correspond to integer partitions of $n$, i.e., $\lambda_1 + \hdots+ \lambda_d = n$ and we denote $\lambda \vdash n$. That is,
\begin{equation}
    \phi_n^d(U) \cong \bigoplus_{\lambda  \vdash n \atop \operatorname{len}(\lambda) \leq d} \mathbf{q}^d_{\lambda}(U) \otimes \id_{m_\lambda}, \forall U \in \mathbf{U}_d,
\end{equation}
where $m_\lambda$ is the multiplicity of the polynomial irrep $\mathbf{q}^d_{\lambda}: \mathbf{U}_d \mapsto \mathrm{GL}(\mathbb{C}^{d_\lambda})$. The dimension $d_\lambda$ of $\mathbf{q}^d_{\lambda}$ is equal to the number of \textit{semistandard Young tableaux} (SSYT) of shape $\lambda$, which are obtained by filling numbers from $[d]$ into the Young diagram $\lambda$ such that, the entries strictly increase downwards along each column and weakly increase rightwards along each row. The multiplicity $m_\lambda$ is equal to the number of \textit{standard Young tableaux} (SYT) of shape $\lambda$, which are obtained by filling numbers from $[n]$ into the Young diagram such that, the entries strictly increase along every row and column. The character of $\mathbf{q}^d_\lambda(U)$ may be expressed as a so-called Schur polynomial in terms of $U$'s entries, which is the same as the generating function of an SSYT of shape $\lambda$~\cite{stembridge1987rational}.

The reason that representation theories of $\mathbf{U}_d$ and $\mathfrak{S}_n$ often go hand in hand together is due to the Schur-Weyl duality. We recall some more definitions before stating it. An \textit{algebra} is a vector space equipped with a bilinear product. For example, the \textit{group algebra} of a group $G$ is $\mathbb{C}G = \operatorname{span}_\mathbb{C} \{\ket{g}, g \in G\}$, and the bilinear product follows by extending by linearity the group action of $G$. Let $\mathrm{End}(V)$ denote the set of all complex matrices over a vector space $V$. A \textit{matrix algebra} $\mathcal{A}$ over $V$ is a subspace $\mathrm{End}(V)$ that is closed under matrix multiplication.

\begin{definition}[Centralizer algebra] The centralizer algebra of a matrix algebra $\mathcal{A}$ is the set of all matrices in $\mathrm{End}(V)$ that commute with $\mathcal{A}$, $\mathrm{End}_\mathcal{A}(V) = \{B: [A,B]=0, \forall A \in \mathcal{A} \}$.
\end{definition}

The Schur-Weyl duality asserts that the centralizer algebra of the tensor representation algebra, i.e., $\mathcal{U}_n^d \triangleq\operatorname{span}_{\mathbb{C}}\left\{\phi_n^d(U): U \in \mathbf{U}_d\right\}$, is exactly another algebra coming from a natural representation of the symmetric group $\mathfrak{S}_n$, and vice versa. In particular, consider the following representation of $\mathfrak{S}_n$ that permutes the subsystems in $(\mathbb{C}^d)^{\otimes n}$:
\begin{equation}
    \psi_n^d(\pi) (\ket{i_1}\otimes \hdots \otimes \ket{\psi_n})= \ket{i_{\pi^{-1}(1)}}\otimes \hdots \otimes \ket{\psi_{\pi^{-1}(n)}}.
    \label{eq:psi_n^d}
\end{equation}
Let $\mathcal{A}_n^d \triangleq \operatorname{span}_{\mathbb{C}}\left\{\psi_n^d(\pi): \pi \in \mathfrak{S}_n \right\}$.

\begin{theorem}[Schur-Weyl duality] $\mathcal{U}_n^d$ is the centralizer algebra of $\mathcal{A}_n^d$ and vice versa. In other words, there exists a unitary transformation under which both representations are decomposed into their irreps
\begin{equation}
    \phi_n^d(U) \cong \bigoplus_{\lambda  \vdash n \atop \operatorname{len}(\lambda) \leq d} \mathbf{q}^d_{\lambda}(U) \otimes \id_{m_\lambda}, \qquad \text{and} \qquad\psi_n^d(\pi) \cong \bigoplus_{\lambda  \vdash n \atop \operatorname{len}(\lambda) \leq d} \id_{d_\lambda} \otimes \mathbf{p}_{\lambda}(\pi). 
\end{equation}
\label{thm:schur-weyl}
\end{theorem}
Note, that $\mathbf{p}_{\lambda}$ has no superscript labeling the system size $d$ since they are irreps of $\mathfrak{S}_n$ and thus are independent of $d$ (whereas $\psi_n^d$ does depend on $d$).
Bacon, Chuang, and Harrow \cite{bacon2007quantum} gave the first efficient (polynomial in $n$ and $d$) quantum circuit for this decomposition (the quantum Schur transform). Krovi~\cite{krovi2019efficient} gave another construction that is efficient in $n$ and $\log d$ and is based on the the quantum Fourier transform of the symmetric group.

\subsection{Staircases and rational tableaux}\label{sec:staircase}
In this work, we are interested in a strict generalization of $U^{\otimes n}$ called the \textit{mixed tensor representation}: $U^{\otimes n} \otimes \bar{U}^{\otimes m}$. This is a rational representation of $\mathbf{U}_d$.
As a simple example, consider the dual representation $\Bar{U}$ of $\mathbf{U}_2$. For any unitary matrix $U=\begin{pmatrix}
       a & b \\
       c & d
    \end{pmatrix} \in \mathbf{U}_2$, we have that
\begin{equation}
    \Bar{U} = (U^{-1})^\top= \frac{1}{ad-bc} \begin{pmatrix}
       d & -c \\
       -b & a
    \end{pmatrix}
\end{equation}
The entries of this representation are rational functions of the entries of $U$, so $\Bar{U}$ is a \emph{rational} representation of $\mathbf{U}_2$. 

Notice that for $\mathbf{U}_2$, $\bar{U}$ can be thought of as a composition of the one-dimensional irrep $\mathrm{det}^{-1}$ and the (polynomial) defining irrep $\mathbf{q}_{\square}^2$ because we can write
\begin{equation}
     \bar{U} = \operatorname{det} U^{-1} \begin{pmatrix}
       0 & -1 \\
       1 & 0
    \end{pmatrix} U \begin{pmatrix}
       0 & -1 \\
       1 & 0
    \end{pmatrix}^\dagger.
    \label{eq:dual-defining}
\end{equation}
More generally, Schur~\cite{schur1927rationalen} showed that all rational irreps of $\mathbf{U}_d$ (and of $\mathbf{GL}_d$) are of the form
\begin{equation}
    U \rightarrow \operatorname{det} U^{r} \rho_\lambda(U),
\end{equation}
where $r \in \mathbb{Z}$ and $\lambda$ is a Young diagram with $\mathrm{len}(\lambda) \leq d$. It can be shown that $\operatorname{det}^r \rho_\lambda$ and $\operatorname{det}^s \rho_\mu$ are equivalent if and only if $(\lambda_1 + r, \lambda_2 +r, \hdots ) =  (\mu_1 + s, \mu_2 +s, \hdots )$ (see, e.g., Theorem 7.1 in~\cite{stembridge1987rational}). This means any polynomial irrep ($r \geq 0$) can be labeled only by a partition $\lambda$. Furthermore, adding or removing a full column of boxes from the
diagram of $\lambda$ corresponds to multiplying or dividing the corresponding representation by
$\operatorname{det} U$. The case when $r$ is a large negative integer such that subtracting it from $\lambda$ makes one of the entries negative is exactly the case of rational irreps. We can label the irreps $T^{r,\lambda}$ is via \textit{staircases}, termed by~\cite{stembridge1987rational}. These are simply sequences of weakly decreasing integers
\begin{equation}
    (\gamma_1,\gamma_2, \hdots, \gamma_d) \in \mathbb{Z} \qquad \text{s.t. } \gamma_1 \geq \gamma_2 \geq \hdots \geq \gamma_d,
\end{equation}
where $\gamma_i = \lambda_i - r$ and extra $0$'s are padded if needed. 

Combinatorics with rational irreps can in principle be reduced to well-known methods for polynomial irreps by adding full columns of boxes (equivalently multiplying the rational irreps with the determinant representations). This approach, however, turns out to quickly become unwieldy. For this reason, Stembridge~\cite{stembridge1987rational} developed several generalizations of Young tableaux to staircases.

There is another formalism in~\cite{stembridge1987rational} of staircases that will be useful in this work. A staircase can be seen as a pair of Young diagrams by splitting the positive and negative entries as illustrated in the following figure.
\begin{figure}[H]
    \centering
\includegraphics[width=0.4\textwidth]{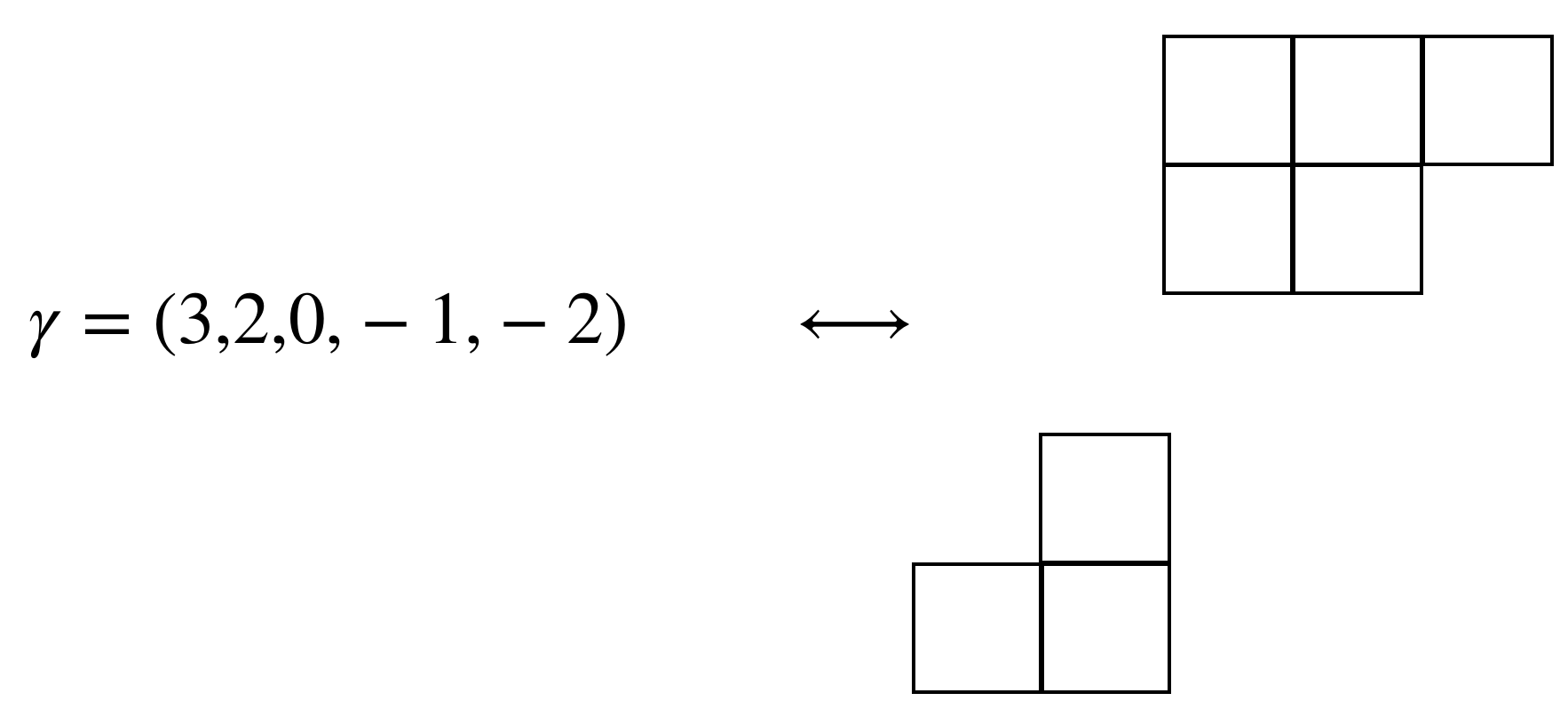}
    \caption{A staircase is split into two Young diagrams.}
    \label{fig:rational}
\end{figure}
We will alternate viewing a rational irrep $\gamma$ as either a tuple of weakly decreasing integers or a pair of Young diagrams $\gamma=[\alpha, \beta]$ depending on context (rather than by the pair $r, \lambda$). Polynomial irreps constitute the special case when $\beta=\varnothing$, the null diagram.

\begin{prop} All rational irreps of $\mathbf{U}_d$ (and more generally, $\mathbf{GL}_d$) can be indexed by staircases $\gamma=[\alpha,\beta]$ of length $d$ such that $\operatorname{len}(\alpha)+\operatorname{len}(\beta)  \leq  d$.
\end{prop}

Filling the boxes with numbers following certain rules gives rise to rational tableaux \cite{stembridge1987rational, kwon2008rational} whose definition is given below for completeness.

\begin{definition}[Rational (semistandard) tableaux] The $\gamma  = [\alpha,\beta]$ be a staircase of length $d$. A \emph{rational tableau} of shape $\gamma$ is a pair $[S,T]$ of semistandard Young tableaux with elements $\leq d$ such that
\begin{enumerate}
    \item $S$ if of shape $\alpha$, and $T$ is of shape $\beta$
    \item $|\{j:S(j,1) \leq  i\}|+|\{j:T(j,1)\leq i\}| \leq i$ for $1 \leq i  \leq d$,
\end{enumerate}
where $S(j,k)$ is the number filled in the box at row $j$, column $k$ of $S$.
\end{definition}
Rational tableaux of shape $\gamma$ yield a basis for the irrep labeled by staircase $\gamma$, just as semistandard Young tableaux (SSYT) do for a Young diagram. In this work, however, we will use a basis more convenient for quantum implementation rather than rational tableaux. This basis will be described in Section~\ref{sec:GTbasis}.
The generalization of standard Young tableaux (SYT) for staircases is called \emph{up-down staircase tableaux} and will be described in Section~\ref{sec:bratteli}.

\subsection{Walled Brauer algebra}

Similar to \cite{bacon2007quantum}, the mixed Schur transform construction in this work will rely on a generalization of Schur-Weyl duality. Before describing this generalization in the next subsection, we need to introduce the \emph{walled Brauer algebra} and its representation on $V^{\otimes n} \otimes \bar{V}^{\otimes m}$ where $V \equiv \mathbb{C}^d$. We follow the presentation in~\cite[Sections 3-5]{grinko2022linear}.

Let $n,m > 0$ and $d \in \mathbb{C}$ (in this work we need only consider positive integers $d$). The walled Brauer algebra $\mathcal{B}_{n,m}^d$ consists of complex linear combinations of diagrams, where each diagram has two rows of $n+m$
nodes each, with a vertical “wall” between the first $n$ and the last $m$ nodes. The
nodes are connected into pairs, subject to the restriction:
if both nodes are in the same row, they must be on different sides of the wall, while
if they are in different rows, they must be on the same side of the wall. See the following example.
\begin{figure}[H]
    \centering
\includegraphics[width=0.3\textwidth]{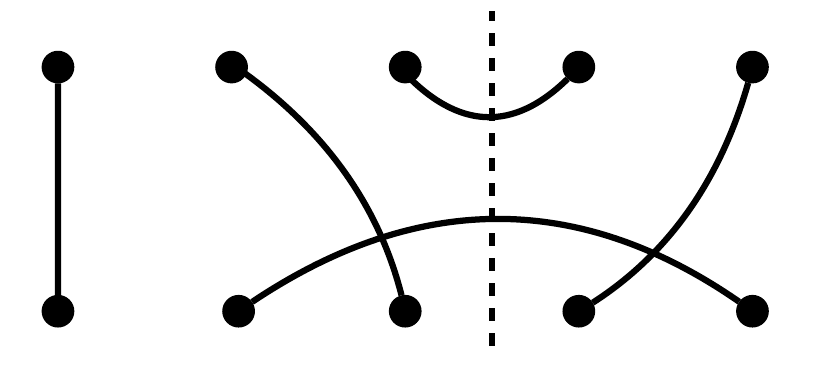}
    \caption{An element in the walled Brauer algebra $\mathcal{B}_{3,2}^d$.}
    \label{fig:walledBrauer}
\end{figure}

We make several remarks. First, the group algebra of $\mathfrak{S}_n \times \mathfrak{S}_m$ is a subalgebra of $\mathcal{B}_{n,m}^d$ without edges crossing the wall. Second, these two algebras become isomorphic when $n=0$ or $m=0$. Third, $\mathcal{B}_{n,m}^d$ itself is a subalgebra of the \textit{full} Brauer algebra, where there are no restrictions on how to pair nodes. Furthermore, there are many interesting properties of $\mathcal{B}_{n,m}^d$ that we will not use in this paper (see \cite{benkart1994tensor, grinko2022linear}). For example, multiplication in $\mathcal{B}_{n,m}^d$ corresponds to concatenating two diagrams vertically, each loop formed in this process yields a factor of $d$. For our purposes, we remark an important connection between the walled Brauer algebra and the symmetric group algebra. That is, as vector spaces, they are related via the partial transpose operation defined in Figure~\ref{fig:partialtranspose}. Therefore, $\operatorname{dim} \mathcal{B}_{n,m}^d = (n+m)!$ and $\mathcal{B}_{n,m}^d$ and $\mathbb{C}\mathfrak{S}_{n+m}$ are isomorphic to each other \emph{as vector spaces} (but not as algebras).

\begin{figure}[H]
    \centering
\includegraphics[width=0.7\textwidth]{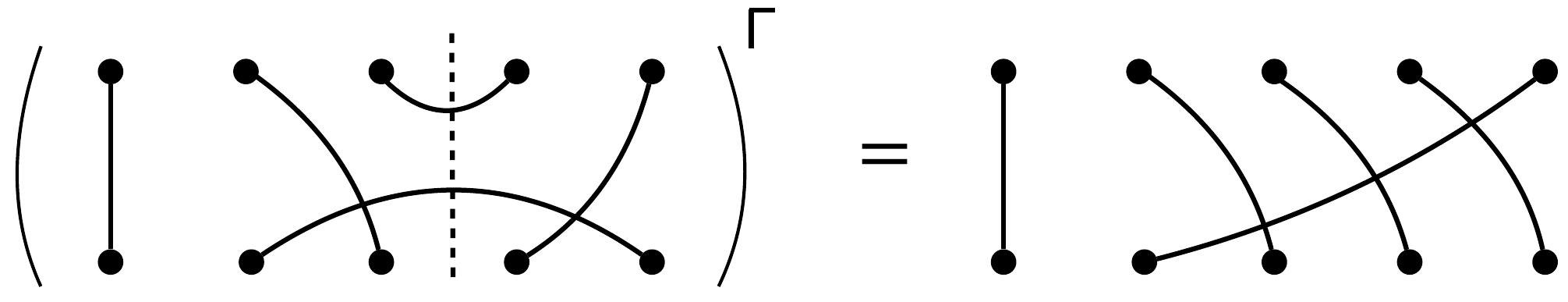}
    \caption{Partial transpose $\sigma^{\ptp}$ of a diagram $\sigma$ in $\mathcal{B}_{n,m}^d$ corresponds to exchanging the last $m$ nodes.}
    \label{fig:partialtranspose}
\end{figure}

Consider the following matrix representation of the walled Brauer algebra $\mathcal{B}_{n,m}^d$. Define, for any diagram $\sigma$,
\begin{equation}
    \psi_{n,m}^d(\sigma) (\ket{i_1}\ket{i_2}\hdots \ket{i_{n+m}}) = \sum_{j_1,\hdots,j_{n+m}=1}^d \sigma_{j_1,\hdots,j_{n+m}}^{i_1,\hdots,i_{n+m}} (\ket{j_1}\ket{j_2}\hdots \ket{j_{n+m}}),
\end{equation}
where $\sigma_{j_1,\hdots,j_{n+m}}^{i_1,\hdots,i_{n+m}}= \prod_{(t,b)\in \sigma} \delta_{i_t,j_b}$ (that is, the variables must agree for all connected pair of nodes in $\sigma$). As an example, for the diagram in Fig. \ref{fig:walledBrauer}, $\sigma_{j_1,\hdots,j_{n+m}}^{i_1,\hdots,i_{n+m}}= \delta_{i_1,j_1} \delta_{i_2,j_3}\delta_{i_3,i_4}\delta_{i_5,j_4}\delta_{j_2,j_5}$. Note that $\psi_{n,m}^d(\sigma)$ is not a unitary operator, but if we consider the partially transposed diagram in Fig. \ref{fig:partialtranspose}, then $\psi_{n,m}(\sigma^{\ptp})$ is a qudit permutation operator. This should be reassuring because we asserted earlier that the walled Brauer algebra $\mathcal{B}_{n,m}^d$ is a partial transpose of the symmetric group algebra $\mathbb{C}\mathfrak{S}_{n+m}$. Indeed, the algebra generated by this representation,
\begin{equation} \mathcal{A}_{n,m}^d\triangleq\operatorname{span}_{\mathbb{C}}\left\{\psi_{n,m}^d(\sigma): \sigma \in \mathcal{B}_{n,m}^d \right\}
\end{equation}
is called the \emph{partially transposed permutation} (PTP) matrix algebra.\footnote{It is worth noting that when $d \geq n+m$, $\psi_{n,m}^d$ is a faithful representation of $\mathcal{B}_{n,m}^d$.} It will play the role of the symmetric group algebra in the mixed Schur-Weyl duality (the Schur-Weyl duality is recovered when $n=0$ or $m=0$).

\subsection{Mixed Schur-Weyl duality}\label{sec:mixedSW}
We can now describe the extension of the Schur-Weyl duality to mixed tensor representations, following
\cite[Section 3]{grinko2022linear}.

Define the mixed tensor representation
\begin{equation}
    \phi^d_{n,m} (U)= U^{\otimes n} \otimes \Bar{U}^{\otimes m}.
\end{equation}
and let us denote by $\mathcal{U}_{n,m}^d$ the algebra generated by it, i.e., 
\begin{equation}
    \mathcal{U}_{n,m}^d \triangleq\operatorname{span}_{\mathbb{C}}\left\{\phi_{n,m}^d(U): U \in \mathbf{U}_d\right\}.
\end{equation}
We will refer to the irreps of $\mathbf{U}_d$ labeled by $\gamma$ as (boldface) $\mathbf{q}_\gamma^d$ and the vector space carrying it (mathcal) $\mathcal{Q}^d_\gamma$. As it turns out, staircases $\gamma$ also label the irreps of $\mathcal{A}_{n,m}^d$. Paralleling the standard Schur-Weyl case, we refer to them as $\mathbf{p}_\gamma^d$ the irrep and as $\mathcal{P}_\gamma^d$ the vector space carrying it. Note that these notations have an extra superscript $d$ now. Indeed, the irreps of $\mathcal{A}^d_{n,m}$ surprisingly depend on $d$, which was not the case for the $\mathfrak{S}_n$ group algebra. We will see an example of this dependence in Section~\ref{sec:bratteli}.

\begin{theorem}[Mixed Schur-Weyl duality \cite{koike1989decomposition, benkart1994tensor}] \hphantom \\
\begin{center}
    $\mathcal{U}^d_{n,m}$ is the centralizer algebra of $\mathcal{A}^d_{n,m}$ and vice versa.
\end{center}
 In other words, there exists a \emph{mixed Schur transform} that decomposes the vector space $V^{\otimes n} \otimes \bar{V}^{\otimes m}$ as 
\begin{equation}
    V^{\otimes n} \otimes \bar{V}^{\otimes m} \cong \bigoplus_{\gamma \in \operatorname{Irr}(\mathcal{U}_{n,m}^d)} \mathcal{Q}^d_\gamma \otimes \mathcal{P}^d_\gamma,
    \label{eq:mixedSW}
\end{equation}
where $\operatorname{Irr}(\mathcal{U}_{n,m}^d)$ is the set of staircases $\gamma$ that \emph{simultaneously} label the $\mathcal{Q}^d_\gamma$ of $\mathbf{U}_d$ irreps constituting in $\mathcal{U}_{n,m}^d$ and the irreps $\mathcal{P}^d_\gamma$ of $\mathcal{A}_{n,m}^d$ constituting in $\mathcal{A}^d_{n,m}$.
\label{thm:mixedSW}
\end{theorem}

From the perspective of $\mathbf{U}_d$, the mixed Schur transform block-diagonalizes the mixed tensor representation,
\begin{equation}
    W_{\mathrm{Sch}(n,m)} (U^{\otimes n} \otimes \Bar{U}^{\otimes m}) W_{\mathrm{Sch}(n,m)}^\dagger = \bigoplus_{\gamma \in \operatorname{Irr}(\mathcal{U}_{n,m}^d) } \mathbf{q}^d_\gamma (U) \otimes \id_{m_\gamma}, \forall U \in \mathbf{U}_d,
\end{equation}
where $m_\gamma$ is the multiplicity of $\mathbf{q}_\gamma^d$ which is equal to the dimension of the irrep $\mathbf{p}^d_\gamma$ of $\mathcal{A}^d_{n,m}$.

Dually, the mixed Schur transform block-diagonalizes the PTP matrix algebra
\begin{equation}
    W_{\mathrm{Sch}(n,m)} \psi_{n,m}^d(\sigma) W_{\mathrm{Sch}(n,m)}^\dagger = \bigoplus_{\gamma \in \operatorname{Irr}(\mathcal{A}_{n,m}^d) } \id_{d_\gamma} \otimes \mathbf{p}^d_\gamma (\sigma),  \forall \sigma \in \mathcal{B}_{n,m}^d,
\end{equation}
where $d_\gamma$ is the multiplicity of $\mathbf{p}_\gamma^d$ which is equal to the dimension of the irrep $\mathbf{q}^d_\gamma$ of $\mathbf{U}_d$.

To implement the mixed Schur transform on a quantum computer, we need to choose bases for both the $\mathcal{Q}^d_\gamma$ and $\mathcal{P}^d_\gamma$ registers and a corresponding labeling scheme. As described in the next subsections, we will choose certain canonical bases $\{\ket{q_\gamma}\}$ for $\mathcal{Q}^d_\gamma$ and $\{\ket{p_\gamma}\}$ for $\mathcal{P}^d_\gamma$. Doing so defines a matrix representation of the mixed Schur transform. In particular, a vector in this canonical basis is a linear combination of computational basis elements
\begin{equation}
    \ket{\gamma, q_\gamma,p_\gamma}_\mathrm{Sch} = (W^{\gamma,q_\gamma,p_\gamma}_{i_1\hdots i_{n+m}})^*  \ket{i_1\hdots i_n i_{n+1}\hdots i_{n+m}},
\end{equation}
whose coefficients $(W^{\gamma,q_\gamma,p_\gamma}_{i_1\hdots i_{n+m}})^*$ define the entries of the mixed Schur transform.

Our goal will be to construct the mixed Schur transform as an isometry $W_{\mathrm{Sch}(n,m)}$ on a quantum computer which turns such canonical basis states into the bitstring description of them
\begin{equation}
    W_{\mathrm{Sch}(n,m)} \ket{\gamma, q_\gamma,p_\gamma}_\mathrm{Sch} =  \ket{\gamma, q_\gamma,p_\gamma},
    \label{eq:def-mixedschur}
\end{equation}
where the registers on the RHS register stores \textit{bit strings} describing the irrep label and basis elements, and the LHS, $\ket{\cdot}_{Sch}$, is the corresponding canonical basis state described by the RHS bit strings. The mixed Schur-Weyl duality implies that this isometry will also be the irrep decomposition of the partially transposed permutation matrix algebra.

\begin{remark}
     The number of irreps show up in the mixed Schur transform in Eq.~\eqref{eq:mixedSW}, i.e., the size $|\operatorname{Irr}(\mathcal{U}_{n,m}^d)|=|\operatorname{Irr}(\mathcal{A}_{n,m}^d)|$ is given by the following formula~\cite[Equation 3.28]{bulgakova2020some}
    \begin{equation}
        \sum_{k=0}^{\min (n,m)} f^d_{n-k,m-k},
    \end{equation}
where $f^d_{n-k,m-k}$ is the number of staircases $\gamma=[\alpha,\beta]$ of length $\operatorname{len}(\gamma)\leq d$ such that $\alpha \vdash n-k$ and $\beta \vdash m-k$. However, there is no ``one-letter'' formula for this in terms of $n,m,$ and $d$.
For example, when $m=0$ and $d \geq n$, this number equals the number of partitions of an integer $n$, which is an open problem in number theory. A crude upper bound is of course $\min \{d^{n+m}, (n+m)!\}$.
\label{remark:countirreps}
\end{remark}

\subsection{Gelfand-Tsetlin basis for the $\mathcal{Q}^d_\gamma$ register}\label{sec:GTbasis}
For the register $\mathcal{Q}^d_\gamma$, we will use the Gelfand-Tsetlin (GT) basis for rational irreps (see \cite[Section 2]{louck1970recent}). We summarize the main ideas of this scheme in the sequel. Suppose $\mathcal{Q}^d_\gamma$ is an irrep of $\mathbf{U}_d$. What if we view it as a representation of $\mathbf{U}_{d-1}$? In this case, it becomes a reducible representation and thus can be decomposed into $\mathbf{U}_{d-1}$ irreps. The Weyl branching rule asserts that, for any $\mathcal{Q}^d_\gamma$, the $\mathbf{U}_{d-1}$ irreps that appear in this decomposition are $\mathcal{Q}^{d-1}_{\mu}$ such that $\mu$ are staircases that satisfy the ``betweenness'' condition
\begin{equation}
    \gamma_1 \geq \mu_1 \geq \gamma_2 \geq \hdots \geq \mu_{d-1} \geq \gamma_d.
\end{equation}
We say $\mu$ \emph{interlaces} $\gamma$ and denote $\mu \precsim \gamma$. Furthermore, this decomposition is multiplicity-free, i.e., each such $\mu$ appears exactly once. This implies that any element $\ket{q_\gamma}$ in $\mathcal{Q}^d_\gamma$ can be labeled by $\ket{\mu,q_\mu}$ for some $\mu \precsim \gamma$ and $\ket{q_\mu} \in \mathcal{Q}^{d-1}_\mu$. Denote $\mu^{(d-1)}\equiv \mu$ to make the fact that $\mu$ is an irrep of $\mathbf{U}_{d-1}$ explicit. Recursively applying the Weyl branching rule all the way to $\mathcal{U}_1$, we find that $\ket{q_\gamma}$ can be labelled by
\begin{equation}
    \ket{q_\gamma}=\ket{\mu^{(d-1)}, q_{\mu^{(d-1)}}} = \ket{\mu^{(d-1)}, \mu^{(d-2)}, q_{\mu^{(d-2}}} = \hdots = \ket{\mu^{(d-1)}, \mu^{(d-2)},\hdots, \mu^{(1)}, q_{\mu^{(1)}}}.
\end{equation}
Note that we can omit the register $\ket{q_{\mu^{(1)}}}$ since it is one-dimensional. This is the Gelfand-Tsetlin (GT) basis that we will use and is an example of so-called subgroup-adapted bases~\cite{moore2006generic}. The \textit{Gelfand pattern} is a triangular array that collects all intermediate irreps:

\begin{equation}
    \begin{array}{ccccccccc}
         \mu^{(d)}_1& &\mu^{(d)}_2 & &\hdots& & \mu^{(d)}_{d-1} & & \mu^{(d)}_d   \\
         & \mu^{(d-1)}_1 & & \mu^{(d-1)}_2 & & \hdots & &  \mu^{(d-1)}_{d-1} & \\
         & & & & \hdots & & & & \\
         & & & \mu^{(2)}_1 &  & \mu^{(2)}_2 & & & \\
         & & & & \mu^{(1)}_1 & & & &
    \end{array}.
\end{equation}
A basis state $q_\gamma$ of $\gamma \equiv \mu^{(d)}$ can be uniquely labeled by a Gelfand pattern.

The dimension of $\mathcal{Q}^d_{\gamma}$ is given by the Weyl dimension formula~\cite[Eq. (2.51)]{louck1970recent}
\begin{equation}
    \operatorname{dim} \mathcal{Q}^d_{\gamma} = \frac{\prod_{i < j} (\gamma_i -\gamma_j-i+j)}{(d-1)!}.
    \label{eq:weyl-dim}
\end{equation}
Evidently from this formula, adding a full column of boxes to $\gamma$ (equivalent to multiplying the irrep with the determinant irrep) does not change the dimension of the irrep.

In the remaining of the paper, the notation $q_\gamma \in \mathcal{Q}_\gamma^d$ means $\ket{q_\gamma}$ is a GT-basis element in $\mathcal{Q}_\gamma^d$.

\subsection{Basis for $\mathcal{A}^d_{n,m}$ irreps from its Bratteli diagram}\label{sec:bratteli}
For the register $\mathcal{P}^d_\gamma$, we use an orthogonal basis for the irreps of $\mathcal{A}_{n,m}^d$ that can be represented by rational tableaux.
It is constructed based on a multiplicity-free tower of subalgebras of $\mathcal{B}_{n,m}^{d}$.

\begin{definition}[Multiplicity-free family of algebras, paraphrased from \cite{doty2019canonical}] A family $\mathcal{A}_0, \hdots, \mathcal{A}_n$ of finite-dimensional
reducible algebras over $\mathbb{C}$ is multiplicity-free if the following requirements hold
    \begin{enumerate}
        \item (Triviality) $\mathcal{A}_0 \cong \mathbb{C}$.
        \item (Embedding) For each $k$, there is a unity-preserving algebra embedding $\mathcal{A}_{k-1} \hookrightarrow \mathcal{A}_{k}$.
        \item (Branching) The restriction to $\mathcal{A}_{k-1}$ of an $\mathcal{A}_k$-irrep $V$ is isomorphic to a direct sum of pairwise inequivalent $\mathcal{A}_{k-1}$-irreps.
    \end{enumerate}
\end{definition}
It turns out that the walled Brauer algebras form a multiplicity-free family via the following chain of embeddings
\begin{equation}
    \mathbb{C} \cong \mathcal{B}^d_{0,0} \hookrightarrow \mathcal{B}^d_{1,0} \hookrightarrow \hdots \hookrightarrow\mathcal{B}^d_{n,0} \hookrightarrow \mathcal{B}^d_{n,1} \hookrightarrow 
    \hdots
    \hookrightarrow \mathcal{B}^d_{n,m}.
    \label{eq:tower}
\end{equation}

\begin{figure}
\centering\includegraphics[width=0.65\textwidth]{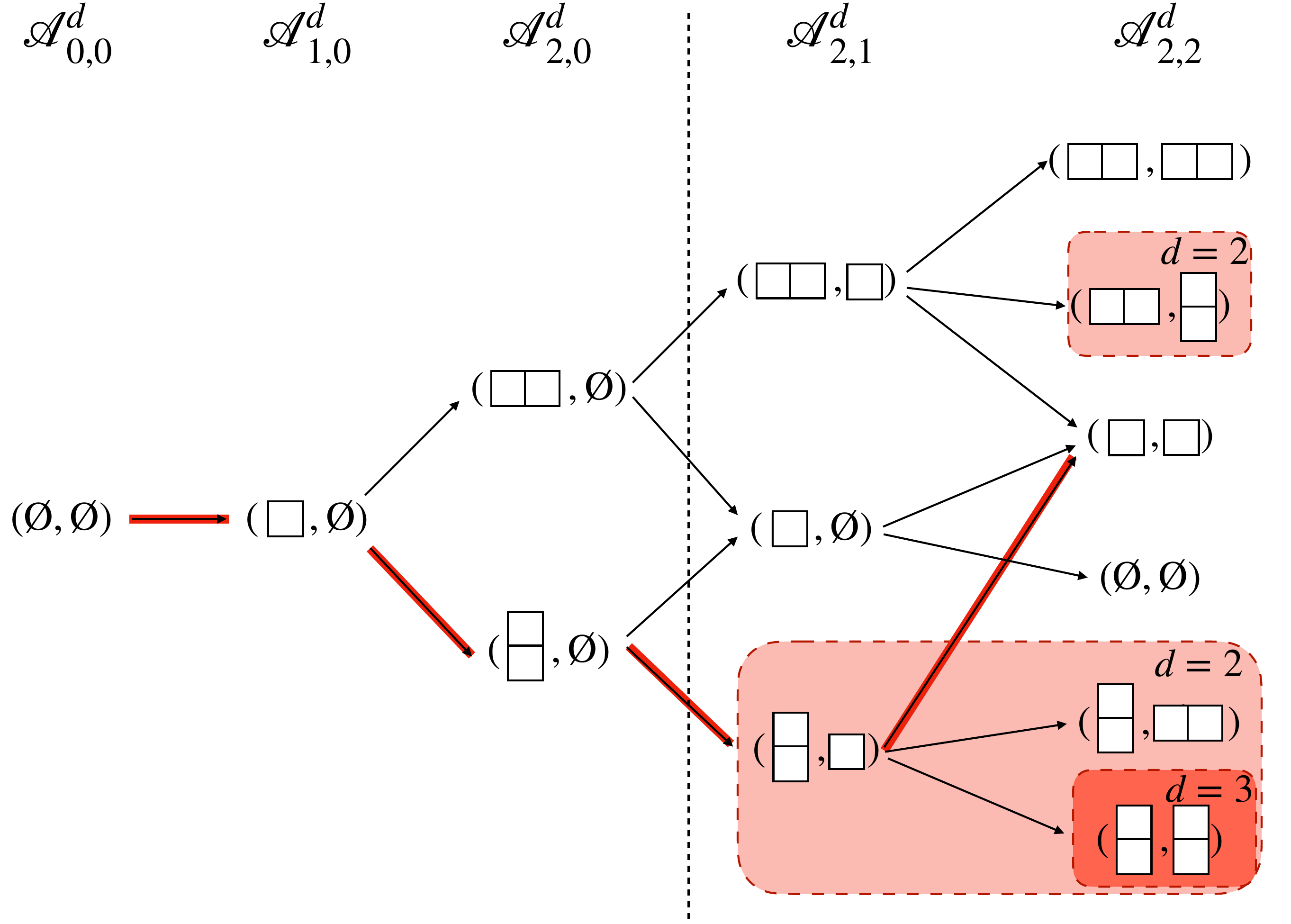}
    \caption{(Reproduction of Figure 3 in~\cite{grinko2022linear}) The Bratteli diagram associated with the multiplicity-free tower of subalgebras $\mathbb{C} \cong \mathcal{A}^d_{0,0} \hookrightarrow \mathcal{A}^d_{1,0} \hookrightarrow \mathcal{A}^d_{2,0} \hookrightarrow \mathcal{A}^d_{2,1} \hookrightarrow \mathcal{A}^d_{2,2}$, which is obtained by removing staircases of length larger than $d$ (e.g., red regions) from the Bratteli diagram of the walled Brauer algebra Eq. \eqref{eq:tower}. Note that  if $d \geq n+m$, then the Bratteli diagrams of $\mathcal{A}_{n,m}^d$ and $\mathcal{B}_{n,m}^d$ coincide. ``Multiplicity-free'' means that there is at most one edge between two vertices in adjacent layers.}
    \label{fig:bratteli}
\end{figure}

This chain can be described via a Bratteli diagram as shown in Fig. \ref{fig:bratteli}, where layers from $0$ to $n$ consist of irreps of $\mathcal{B}^d_{0,0}$ to $\mathcal{B}^d_{n,0}$, and layers $n+1$ onwards consist of irreps of $\mathcal{B}^d_{n,1}$ to $\mathcal{B}^d_{n,m}$. The rules of adding edges are as follows
\begin{itemize}
    \item If to the left side of the wall, then add one box to diagram $\alpha$ in staircase $[\alpha,\varnothing]$.
    \item If to the right side of the wall, then either
    \begin{itemize}
        \item add one box to diagram $\beta$ in staircase $[\alpha,\beta]$
        \item or remove one box from diagram $\alpha$ in staircase $[\alpha,\beta]$. 
    \end{itemize}
\end{itemize}

Given a staircase $\gamma=[\alpha,\beta]$, we denote by $\gamma + \square$ the set of all possible staircases resulting from either adding one box to $\alpha$ or removing one box from $\beta$ (rule 1) and by $\gamma - \square$ all possible staircases resulting from either removing one box from $\alpha$ or adding one box to $\beta$ (rule 2). Note that in the Bratteli diagram branching rules described earlier, we only ever add boxes to $\alpha$ when applying rule 1 for $\beta=\varnothing$ to the left side of the wall. The Bratteli diagram of the symmetric group $\mathfrak{S}_n$ is a special case when only rule 1 is used (see Figure 1 in \cite{grinko2022linear}).

We now discuss how the previous paragraphs describe a generalization of standard Young tableaux (SYT) from Young diagrams to staircases. A SYT is a Young diagram with $n$ boxes filled with numbers in $[n]$ such that, the entries strictly increase along every row and column. However, an equivalent definition of a SYT is a path in the Bratteli diagram of a $\mathfrak{S}_n$ (see, e.g., Equation 26~\cite{bacon2007quantum}). In this sense, paths in the Bratteli diagram of $\mathcal{B}_{n,m}^d$ generalize SYT. In fact, the generalization is more general and termed up-down staircase tableaux by Stembridge~\cite{stembridge1987rational}.

\begin{definition}[Definition 4.4~\cite{stembridge1987rational}] An up-down staircase tableau $T$ of length $n$ and shape $\gamma$ is a sequence $\varnothing, \gamma^{(0)},\gamma^{(1)}\hdots,\gamma^{(n)}=\gamma$ of staircases each of length $d$ in which $\gamma^{(i)} \in \gamma^{(i-1)} + \square$ or $\gamma^{(i)} \in \gamma^{(i-1)} - \square$ for $1 \leq i \leq n$. The tableau $T$ is said to be of type $\varepsilon = (\varepsilon_1,\hdots,\varepsilon_n)$, where $\varepsilon_i=1$ if $\gamma^{(i)} \in \gamma^{(i-1)} + \square$ (rule 1) and $\varepsilon_i=-1$ otherwise (rule 2).
\end{definition}

For example, the sequence
\begin{equation}
    [\varnothing,\varnothing] \rightarrow [\varnothing,\yd[1]{1}] \rightarrow [\yd[1]{1},\yd[1]{1}] \rightarrow [ \yd[1]{1},\varnothing] \rightarrow [\yd[1]{1,1},\varnothing] \rightarrow [\yd[1]{1,1},\yd[1]{1}]
\end{equation}
is an up-down staircase tableau (UDST) of length $n=5$, type $\varepsilon=(-1,1,1,1,-1)$ and shape $\gamma=[\yd[1]{1,1},\yd[1]{1}]$. Evidently, SYT are special cases when $\varepsilon=(1,1,\hdots,1)$. The term `up-down' follows from similarity with up-down tableaux for symplectic groups.

\begin{remark}
\label{remark:insertionrule}
Stembridge~\cite[Section 3]{stembridge1987rational} proved that given a staircase $\gamma$, the number of UDSTs of a type $\varepsilon$, denoted as $c_{\gamma,\varepsilon}$, only depends on the number of $+1$'s and $-1$'s in $\varepsilon$. In other words, the order in which rule 1 and rule 2 are applied does not matter. Therefore, the number of paths from the root to $\gamma$ in the Bratteli diagram in Fig. \ref{fig:bratteli} is equal to $c_{\gamma,\varepsilon}$ for all type $\varepsilon$ with $n$ $+1$'s and $m$ $-1$'s.
\end{remark}

Furthermore, $c_{\gamma,\varepsilon}$ counts exactly the dimension of $\mathcal{P}^d_\gamma$ in the mixed Schur transform (Eq. \eqref{eq:mixedSW}). This makes it clear why the ordering of applications of rule 1 and rule 2 does not matter: Permuting the factors in $V^{\otimes n} \otimes \bar{V}^{\otimes m}$ should not change its irrep structure.

\begin{lemma}[Corollary 4.7 in \cite{stembridge1987rational}] Let $\gamma$ be a staircase of length $n+m$ and $\varepsilon=(1^n,(-1)^m)$. The multiplicity of $\mathcal{Q}^d_\gamma$, which equals $\operatorname{dim} \mathcal{P}^d_\gamma$, in the decomposition of $V^{\otimes n} \otimes \Bar{V}^{\otimes m}$ is $c_{\gamma,\varepsilon}$. 
\end{lemma}

With this, we can now describe the subalgebra-adapted basis for $\mathcal{A}_{n,m}^d$ irreps. It is obtained from $\mathcal{B}_{n,m}^d$'s Bratteli diagram by removing all vertices that violate $\operatorname{len}(\gamma) \leq d$. The dimension of an irrep $\mathcal{P}^d_\gamma$ in $\mathcal{A}_{n,m}^d$ is equal to the number of paths in the Bratteli diagram from the root $[\varnothing,\varnothing]$ to $\gamma$ at layer $n+m$. Thus, these paths serve as a basis for the irrep $\mathcal{P}^d_\gamma$. For each staircase $\gamma=[\alpha,\beta]$ present in mixed Schur decomposition, we will label an orthogonal basis for $\mathcal{P}^d_{\gamma}$ by descriptions of paths in the Bratteli diagram of $\mathcal{A}_{n,m}^{d}$ starting from the root $[\varnothing,\varnothing]$ to $\gamma$. A simple scheme of indexing which boxes get added or removed along the path uses $\mcO((n+m) \log d)$ bits. In the remaining of the paper, the notation $p_\gamma \in \mathcal{P}_\gamma^d$ means $\ket{p_\gamma}$ is a basis element in $\mathcal{P}_\gamma^d$ defined by paths in the Bratteli diagram.

Notice that irreps labeled by the same staircase $\gamma$ can have different dimensions depending on $d$. According to Fig.~\ref{fig:bratteli}, the vertex $[\square,\square]$, as an irrep of $\mathcal{A}_{2,2}^2$, has dimension three because there are three paths from the root to it. However, as an irrep of $\mathcal{A}_{2,2}^3$, it has dimension four due to the red path Fig.~\ref{fig:bratteli} becoming valid when increasing $d=2$ to $d=3$. This is why we need the superscript $d$ to denote the irreps $\mathcal{P}^d_\gamma$.

As an example, we count the dimensions and multiplicities of all irreps when $n=2,m=2,d=3$ using the Bratteli diagram. According to Fig.~\ref{fig:bratteli}, there are five irreps appearing in the mixed Schur decomposition Eq.~\eqref{eq:mixedSW}. Let us verify that the dimensions match up. The staircase $[\yd[1.0]{2},\yd[1.0]{2}]$, as an irrep of $\mathcal{U}_3$, has dimension (dim.) 27 by applying Eq.~\eqref{eq:weyl-dim} and multiplicity (mult.) 1 (since there is a single path from the root $[\varnothing,\varnothing]$ to it in Fig.~\ref{fig:bratteli}). Similarly, $[\yd[1.0]{2},\yd[1.0]{1,1}]$ has dim. 10 and mult. 1; $[\yd[1.0]{1},\yd[1.0]{1}]$ has dim. 8 and mult. 4; $[\varnothing, \varnothing]$ has dim. 1 and mult. 2; $[\yd[1.0]{1,1},\yd[1.0]{2}]$ has dim. 10 and mult. 1. Indeed, $27\cdot 1 + 10 \cdot 1 + 8 \cdot 4 + 1 \cdot 2 + 10 \cdot 1  =3^4$.

Let us contrast the irrep structure of $\mathcal{A}^3_{2,2}$ to that of the symmetric group algebra $\mathcal{A}^3_{4}$. According to the $\mathfrak{S}_4$ Bratteli diagram in~\cite[Fig. 1]{grinko2022linear}, we find that $U^{\otimes 4}$ contains four irreps labeled by the Young diagrams $\yd[1.0]{4}$ (dim. 15, mult. 1), $\yd[1.0]{3,1}$ (dim. 15, mult. 3), $\yd[1.0]{2,2}$ (dim. 6 , mult. 2), and $\yd[1.0]{2,1,1}$ (dim. 3, mult. 3).  We can see that the block-structural difference between $\mathcal{U}_{n,m}^d$ ($\mathcal{A}^d_{n,m}$) and $\mathcal{U}_{n+m}^d$ ($\mathcal{A}^d_{n+m}$) already manifests in this small example.

\begin{remark} While the number of SYT of shape $\lambda$, which is equal to the multiplicity of the polynomial irrep $\mathcal{Q}^d_\lambda$ labeled by Young diagram $\lambda$, is provided by hook length formula (see Equation 28, \cite{bacon2007quantum}). To our knowledge, there is no known explicit formula for the number of up-down staircase tableaux (which equals the multiplicity of rational irreps in the mixed Schur-Weyl duality). There are explicit formulas in special cases, such as when $n=m$ or when $d \geq n+m$ (see \cite[Proposition 4.8]{stembridge1987rational}). However, this is not an issue for the quantum implementation since we only need an upper bound on $\operatorname{dim} \mathcal{P}^d_\gamma$ for this (we similarly did not have a closed formula for the total number of irreps, even in the normal Schur transform, see Remark \ref{remark:countirreps}). We have the following bound on $\operatorname{dim} \mcP^d_\gamma$ for $\gamma=[\alpha,\beta]$~\cite[Lemma 1.3.3]{bulgakova2020some}
\begin{equation}
    \operatorname{dim} \mcP^d_\gamma \leq \begin{pmatrix}
        n\\
        k
    \end{pmatrix} \begin{pmatrix}
        m\\
        k
    \end{pmatrix} k! \operatorname{dim} \mcP_\alpha \operatorname{dim} \mcP_\beta
\end{equation}
where $k \equiv n-\operatorname{len}(\alpha) = m - \operatorname{len}(\beta)$ and $ \operatorname{dim} \mcP_\alpha ,\operatorname{dim} \mcP_\beta$ can be computed using the hook length formula.
\end{remark}

\subsection{Example of the mixed Schur transform}\label{sec:example}
To illustrate the mixed Schur-Weyl duality and canonical bases developed so far, we give an example of the mixed Schur transform with $n=2,m=1$, and $d=2$. According to Fig.~\ref{fig:bratteli}, there are two irreps in $\Bar{U} \otimes U \otimes U$, which are labeled by staircases $\gamma = [\yd[1]{1},\varnothing]$ and $\mu= [\yd[1]{2},\yd[1]{1}]$. 

The irrep $\mathcal{Q}^2_\mu$ has multiplicity one, corresponding to the path
\begin{equation}
    \mathrm{path}_0 = [\varnothing,\varnothing] \rightarrow [\yd[1]{1},\varnothing] \rightarrow [\yd[1]{2},\varnothing] \rightarrow [\yd[1]{2},\yd[1]{1}].
\end{equation}

The irrep $\mathcal{Q}^2_\gamma$ has multiplicity two, corresponding to the paths
\begin{align}
    &\mathrm{path}_1 = [\varnothing,\varnothing] \rightarrow [\yd[1]{1},\varnothing] \rightarrow [\yd[1]{2},\varnothing] \rightarrow [\yd[1]{1},\varnothing],\\
    &\mathrm{path}_2 = [\varnothing,\varnothing] \rightarrow [\yd[1]{1},\varnothing] \rightarrow [\yd[1]{1,1},\varnothing] \rightarrow [\yd[1]{1},\varnothing].
\end{align}

According to Section~\ref{sec:GTbasis}, the irrep $\mathcal{Q}^2_\mu$ has dimension four and its GT basis elements are labeled by the following Gelfand patterns
\begin{align}
    &\mathrm{GT}^\mu_1 = \begin{pmatrix}
        2& &-1  \\
        &-1&
    \end{pmatrix}, \qquad 
    &\mathrm{GT}^\mu_2 = \begin{pmatrix}
        2& &-1  \\
        &0&
    \end{pmatrix}, \\
    &\mathrm{GT}^\mu_3 = \begin{pmatrix}
        2& &-1  \\
        &1&
    \end{pmatrix}, \qquad 
    &\mathrm{GT}^\mu_4 = \begin{pmatrix}
        2& &-1  \\
        &2&
    \end{pmatrix}.
\end{align}
The irrep $\mathcal{Q}^2_\gamma$ has dimension two and its GT basis elements are labeled by the following Gelfand patterns
\begin{align}
    &\mathrm{GT}^\gamma_1 = \begin{pmatrix}
        1& &0  \\
        &0&
    \end{pmatrix}, \qquad 
    &\mathrm{GT}^\gamma_2 = \begin{pmatrix}
        1& &0  \\
        &1&
    \end{pmatrix}.
\end{align}
Below, we will identify these 8 mixed Schur basis elements with row/column indices as follows
\begin{align}
    &\ket{1} = \ket{\gamma, \mathrm{GT}^\gamma_1, \mathrm{path}_1}, \qquad &\ket{2} = \ket{\gamma, \mathrm{GT}^\gamma_2, \mathrm{path}_1},\\
    &\ket{3} = \ket{\gamma, \mathrm{GT}^\gamma_1, \mathrm{path}_2}, \qquad &\ket{4} = \ket{\gamma, \mathrm{GT}^\gamma_2, \mathrm{path}_2},\\
    &\ket{5} = \ket{\mu, \mathrm{GT}^\mu_1, \mathrm{path}_0}, \qquad &\ket{6} = \ket{\mu, \mathrm{GT}^\mu_2, \mathrm{path}_0},\\
    &\ket{7} = \ket{\mu, \mathrm{GT}^\mu_3, \mathrm{path}_0}, \qquad &\ket{8} = \ket{\mu, \mathrm{GT}^\mu_4, \mathrm{path}_0}.
\end{align}
Following the formulas in Chapter 18.2.10 of~\cite{vilenkin1992representations} (also see Section~\ref{sec:recursive}), we obtain the following mixed Schur transform
\begin{equation}
    W_{\mathrm{Sch}(2,1)}= \begin{pmatrix}
       \frac{1}{\sqrt{2}}& 0 & 0& 0& 0& 0& \frac{1}{\sqrt{2}}& 0 \\
       0& \frac{1}{\sqrt{2}} & 0& 0& 0& 0& 0& \frac{1}{\sqrt{2}} \\
       -\frac{1}{\sqrt{6}} & 0 & 0& 0& 0& -\sqrt{\frac{2}{3}}& \frac{1}{\sqrt{6}}& 0 \\
       0& \frac{1}{\sqrt{6}} & -\sqrt{\frac{2}{3}}& 0& 0& 0& 0& -\frac{1}{\sqrt{6}} \\
       0& 0 & 0& 0& 1& 0& 0& 0 \\
       -\frac{1}{\sqrt{3}}& 0 & 0& 0& 0& \frac{1}{\sqrt{3}}& \frac{1}{\sqrt{3}}& 0 \\
       0& -\frac{1}{\sqrt{3}} & -\frac{1}{\sqrt{3}}& 0& 0& 0& 0& \frac{1}{\sqrt{3}}\\
       0& 0 & 0& 1& 0& 0& 0& 0
    \end{pmatrix}.
\end{equation}
We now verify the mixed Schur transform block-diagonalizes the mixed tensor $\Bar{U}\otimes U \otimes U$ (note that we have moved $\Bar{U}$ to the front to be compatible with the symmetry of the Choi matrix below). Let $U= \begin{pmatrix}
    a & b \\ 
    c & d
\end{pmatrix}$ and $\Bar{U} = (U^{-1})^\top = \frac{1}{ad -bc} \begin{pmatrix}
    d & -b \\
    -c & a
\end{pmatrix}$, we find
\begin{equation}
\begin{aligned}
    &W_{\mathrm{Sch}(2,1)} ( \bar{U} \otimes U \otimes U) W_{\mathrm{Sch}(2,1)}^\dagger \\
    &= \frac{1}{\mathfrak{D}}\begin{pmatrix}
       a\mathfrak{D}& b \mathfrak{D} & 0& 0& 0& 0& 0& 0 \\
       c \mathfrak{D}& d\mathfrak{D} & 0& 0& 0& 0& 0& 0 \\
       0& 0& a\mathfrak{D}& b\mathfrak{D}& 0& 0& 0& 0 \\
       0& 0 & c\mathfrak{D}& d\mathfrak{D}& 0& 0& 0& 0 \\
       0& 0 & 0& 0& a^3& \sqrt{3}a^2 b& \sqrt{3}ab^2& -b^3 \\
       0& 0 & 0& 0& \sqrt{3}a^2 c& a(2bc+ad)& b(bc+2ad)& -\sqrt{3}b^2 d \\
       0& 0 & 0& 0& \sqrt{3}ac^2& c(bc+2ad)& d(2bc+ad)& -\sqrt{3}bd^2 \\
       0& 0 & 0& 0& -c^3& -\sqrt{3}c^2 d& -\sqrt{3}cd^2& d^3
    \end{pmatrix},
\end{aligned}
\end{equation}
where $\mathfrak{D} = ad- bc$.

Next, we verify the mixed Schur transform also block-diagonalizes the centralizer of $ \Bar{U}\otimes U \otimes U$. For concreteness, we consider the set of $1$-to-$2$ qubit unitary-equivariant maps derived in Appendix F.4 of~\cite{nguyen2022theory}. Recall
the Choi state $J^{\mcN} \triangleq (\id \otimes \mcN) (\ket{\Omega}\bra{\Omega})$, where $\ket{\Omega}$ is the normalized EPR state, of a $1$-to-$2$ qubit unitary-equivariant map commutes with $\Bar{U}\otimes U \otimes U$ for any $U \in \mathbf{U}_2$. Ref.~\cite[Appendix F.4]{nguyen2022theory} showed that all such maps that are trace-preserving can be fully parameterized by 4 parameters $\mcN_{t,u,v,w}$
\begin{align}
    \mcN_{t,u,v,w}(\rho) = \frac{\id \otimes \id}{4} &+ \frac{t}{2} (XX + YY + ZZ) +\sum_{P \in \{X,Y,Z\}} \Tr(P \rho) \left( \frac{u}{2} \id \otimes P + \frac{v}{2} P \otimes \id  \right)\\
    & + \frac{w}{2} \left( (YZ-ZY) \Tr(X \rho) + (ZX-XZ) \Tr(Y \rho) + (XY-YX) \Tr(Z \rho) \right).
\end{align}
Under the mixed Schur transform, the corresponding Choi state is found to be
\begin{equation}
    W_{\mathrm{Sch}(2,1)} J^\mcN W_{\mathrm{Sch}(2,1)}^\dagger = \frac{1}{8} \begin{pmatrix}
       A& 0& B& 0& 0& 0& 0& 0 \\
       0& A& 0& B& 0& 0& 0& 0 \\
       C& 0& D& 0& 0& 0& 0& 0 \\
       0& C& 0& D& 0& 0& 0& 0 \\
       0& 0& 0& 0& E& 0& 0& 0 \\
       0& 0 & 0& 0& 0& E& 0& 0 \\
       0& 0 & 0& 0& 0& 0& E& 0 \\
       0& 0 & 0& 0& 0& 0& 0& E
    \end{pmatrix},
\end{equation}
where $A = 1+6u$, $B=-2\sqrt{3}(t+v+4iw)$, $C=-2\sqrt{3}(t+v-4iw)$, $D=1 - 4t - 2u + 4v$, and $E= 1 + 2t - 2u - 2v$. This is indeed in agreement with the irrep structure of $\mathcal{A}^2_{2,1}$.

Therefore, one way to implement these channels (after making sure the parameters are such that the Choi state describes a valid completely positive channel) is to prepare $J^\mcN$ in the mixed Schur basis, then apply the mixed Schur transform to convert it to computational basis, and finally implement the channel \textit{deterministically} as shown in Section~\ref{sec:implementation}.

\section{Quantum circuit for the mixed Schur transform}\label{sec:transform}
In this section, we apply the representation theory tools from Section~\ref{sec:reptheory} to give a quantum implementation of the mixed Schur transform. We start by a brief overview of how the Schur transform was constructed in~\cite{bacon2007quantum} and discuss what new is needed for mixed Schur. Then, we define the notion of the ``dual'' Clebsch-Gordan (CG) transform which is needed to handle the dual irrep $\Bar{U}$. Finally, we describe how the mixed Schur transform can be implemented by combining the (defining) CG and dual CG transforms. We defer the algorithm for the dual CG transform to Section~\ref{sec:recursive}.

\subsection{Schur transform versus mixed Schur transform: an overview}
The main idea of \cite{bacon2007quantum}'s construction is to perform the so-called Clebsch-Gordan (CG) transform sequentially by adding qudits one by one. Let us briefly summarize their construction. In representation theory, the CG transforms take as input a tensor product of two irreps of a group and produces an irrep decomposition. In the Schur transform setting, we are interested in the irreps of $\mathbf{U}_d$, and in particular, when one of the input irreps is the defining irrep $U$ (which is labeled by the Young diagram $(1) \equiv \square$). Due to its commonness, this particular decomposition is often just referred to as the CG transform in quantum information literature. We will follow this and instead refer to the general case as the \textit{generic} CG transform. Under the CG transform, the tensor product of a polynomial irrep $\lambda$ and the defining irrep is well-known to decompose as
\begin{equation}
    \mathcal{Q}^d_\lambda \otimes \mathcal{Q}^d_\square \cong  \bigoplus_{\nu \in \lambda + \square \atop \operatorname{len}(\nu) \leq d} \mathcal{Q}^d_\nu,
    \label{eq:CG}
\end{equation}
where we follow the notation in~\cite[Section 7.1]{harrow2005applications}, denoting by $\lambda + \square$ the set of valid Young diagrams of length at most $d$ resulting from adding one box to $\lambda$. Let the change of basis in the above decomposition be $W^\lambda_\mathrm{CG}$. The CG transform is mathematically defined by $W_\mathrm{CG} = \bigoplus_{\lambda} \ket{\lambda} \bra{\lambda} \otimes W^\lambda_\mathrm{CG}$, where the register $\ket{\lambda}$ `instructs' what decomposition to do and the block sum is over all irreps. Of course, to implement this transform on a quantum computer, we need to perform some cutoff in the block sum over $\lambda$. To achieve the Schur transform from CG transforms, \cite{bacon2007quantum} demonstrate that it is convenient to keep around the input irrep label register as follows, 
\begin{equation}
    W_\mathrm{CG}  \ket{\lambda}\ket{q_\lambda} \ket{i} = \sum_{\nu \in \lambda + \square} \sum_{q_\nu \in \mathcal{Q}^d_\nu} C^{\lambda,\nu}_{q_\lambda,i,q_\nu} \ket{\lambda} \ket{\nu} \ket{q_\nu},
    \label{eq:CGregisters}
\end{equation}
where $q_\lambda \in \mathcal{Q}^d_\lambda, i \in [d]$ and $C^{\lambda,\nu}_{q_\lambda,i,q_\nu}$ are the CG coefficients.
For example, in Fig.~\ref{fig:cascading}, the Schur transform starts by applying $W_\mathrm{CG}$ on $\ket{(\square,\varnothing), i_1,i_2}$ to get output registers of the form $\ket{(\square,\varnothing), \nu, q_\nu}$. Compared to Eq.~\eqref{eq:CGregisters}, $\square$ is identified with input irrep label $\lambda$ (the $\varnothing$ part is only needed for the mixed Schur transform), $i_1$ is identified with $q_\lambda$, and the output irrep $\nu$ takes values from $\square + \square = \{\yd[1]{2}, \yd[1]{1,1}\}$. The next CG transform acts on $\ket{\nu,q_\nu, i_3}$ to yield registers $\ket{\nu, \nu',q_{\nu'}}$ with $\nu' \in \nu + \square$. Continuing this procedure, we obtain the Schur transform. Now, \cite{bacon2007quantum} exploits the fact that keeping track of the $\mathbf{U}_d$ irreps that appear during the sequential Clebsch-Gordan transforms yields an orthogonal basis for the multiplicity space in Eq.~\eqref{eq:schur-weyl}. That is, in Fig. \ref{fig:cascading}, the irrep labels $\ket{\gamma^{(1)}}\hdots \ket{\gamma^{(n-1)}}$ indeed form a basis for the multiplicity space of irrep $\gamma^{(n)}$ (all of these are polynomial irreps, i.e., their second part is $\varnothing$). This is due to the properties of the subgroup-adapted basis and the Schur-Weyl duality, both being special cases of the theory described in Section~\ref{sec:bratteli}.

It is then a natural attempt to generalize~\cite{bacon2007quantum} to the mixed Schur case by constructing a quantum circuit for the ``dual'' CG transform, in which one input irrep is the dual defining representation $\Bar{U}$, and then exploiting the mixed Schur-Weyl duality to label the (mixed) Schur basis. The remaining of the text shows that this can indeed be achieved.

\subsection{Dual Clebsch-Gordan transform}\label{sec:dualCG}
Suppose we are given an irrep $\mathcal{Q}_\gamma^d$ and the dual irrep $\mathcal{Q}^d_{[\varnothing, \square]}$, how does the tensor product $\mathcal{Q}_\gamma^d \otimes \mathcal{Q}^d_{[\varnothing, \square]}$ decompose as a representation of $\mathbf{U}_d$? Following \cite{buhrman2022quantum}, we call this decomposition the \textit{dual} Clebsch-Gordan (CG) transform. Here we give some examples of this irrep decomposition and define a quantum implementation of the dual CG transform which will be given in Section \ref{sec:recursive}.

Recall from Section~\ref{sec:bratteli} the notation $\gamma - \square$ which denotes all valid staircases resulting from either removing one box from $\alpha$ or adding one box to $\beta$ (rule 2). Furthermore, let us denote by $\Gamma_d$ the set of all staircases of length $\leq d$. Paralleling Eq.~\eqref{eq:CG}, it turns out that
\begin{equation}
    \mathcal{Q}_\gamma^d \otimes \mathcal{Q}^d_{[\varnothing, \square]} \cong \bigoplus_{\mu \in \gamma - \square \atop 
    \operatorname{len}(\mu) \leq d} \mathcal{Q}_\mu^d.
\end{equation}
For example, consider the following decomposition in the case of $\mathbf{U}_3$
\begin{equation}
    \left[\yd[1]{3,1}, \yd[1]{2} \right] \otimes [\varnothing,\yd[1]{1}] \cong \left[\yd[1]{2,1}, \yd[1]{2} \right] \oplus \left[\yd[1]{3}, \yd[1]{2} \right] \oplus \left[\yd[1]{3,1}, \yd[1]{3} \right].
\end{equation}
Similar to Eq.~\eqref{eq:CGregisters}, we would like to construct a quantum implementation that effects the dual CG transform. We define it as an isometry $W_\mathrm{dCG}$ the performs the following
\begin{equation}
    W_\mathrm{dCG} \ket{\gamma} \ket{q_\gamma} \ket{i} =\sum_{\gamma' \in \gamma - \square \atop \operatorname{len}(\gamma') \leq d} \sum_{q_{\gamma'} \in \mathcal{Q}^d_{\gamma'}} C^{\gamma,\gamma'}_{q_\gamma,i,q_{\gamma'}}  \ket{\gamma} \ket{\gamma'} \ket{q_{\gamma'}},
\end{equation}
where we will choose $q_\gamma$ or $q_{\gamma'}$ to be the canonical basis elements labeled by Gelfand patterns (Section~\ref{sec:GTbasis}).

Above, we should understand 
$W_\mathrm{dCG}$ as the block-diagonal isometry
\begin{equation}
    W_\mathrm{dCG}= \sum_{\gamma}  \ket{\gamma}\bra{\gamma} \otimes W_\mathrm{dCG}^\gamma,
\end{equation}
which applies the dual CG transforms in parallel conditioned on the input irrep $\gamma$. The register $\ket{\gamma}$ will be chosen sufficiently large to label all relevant irreps appearing in the mixed Schur transform.

The $d=2$ case has been recently solved in~\cite{buhrman2022quantum}. The authors used the fact that the dual irrep $\Bar{U}$, in the case of $\mathbf{U}_2$, can be rewritten as a composition of the inverse determinant irrep $U \rightarrow \operatorname{det} U^{-1}$ and the defining irrep $U$ (up to equivalence), i.e.,
\begin{equation}
     \bar{U} = \operatorname{det} U^{-1} \begin{pmatrix}
       0 & -1 \\
       1 & 0
    \end{pmatrix} U \begin{pmatrix}
       0 & -1 \\
       1 & 0
    \end{pmatrix}^\dagger.
    \label{eq:dual-defining}
\end{equation}
Therefore, the dual CG transform of $\mathbf{U}_2$ is essentially equivalent to the CG transform up to the local unitary $V = \begin{pmatrix}
       0 & -1 \\
       1 & 0
    \end{pmatrix}$.
However, we emphasize that this is only true when $d=2$. This is because, specifically for $\mathbf{SU}_2$ ($\operatorname{det} U =1$), the dual irrep $\bar{U}$ is isomorphic to the defining irrep $U$ via the unitary $V$. Whereas, $\bar{U}$ is generally not isomorphic to $U$ for $\mathbf{SU}_d$ with $d>2$.

For general $d$, there are known algorithms for computing the generic CG transform~\cite{koike1989decomposition, biedenharn1968pattern}. However, a naive construction of the dual CG transform by explicitly computing the entries, as done above, takes time $\mcO(\operatorname{dim} \mathcal{Q}^d_\gamma ) = (n+m)^{\mcO(d^2)}$\footnote{This is because the Gelfand pattern has $d(d+1)/2$ entries whose values may range from $-m$ to $n$.}. This is unfavorable for large $d$. In Section~\ref{sec:recursive}, we present a recursive construction of the dual CG transform via the Wigner-Eckart theorem~\cite{biedenharn1968pattern}, adapting the construction in~\cite{bacon2007quantum}, for a circuit complexity of $\mcO(d^4 \polylog(n,m,d,1/\varepsilon))$.

\begin{figure}
\centering\includegraphics[width=0.8\textwidth]{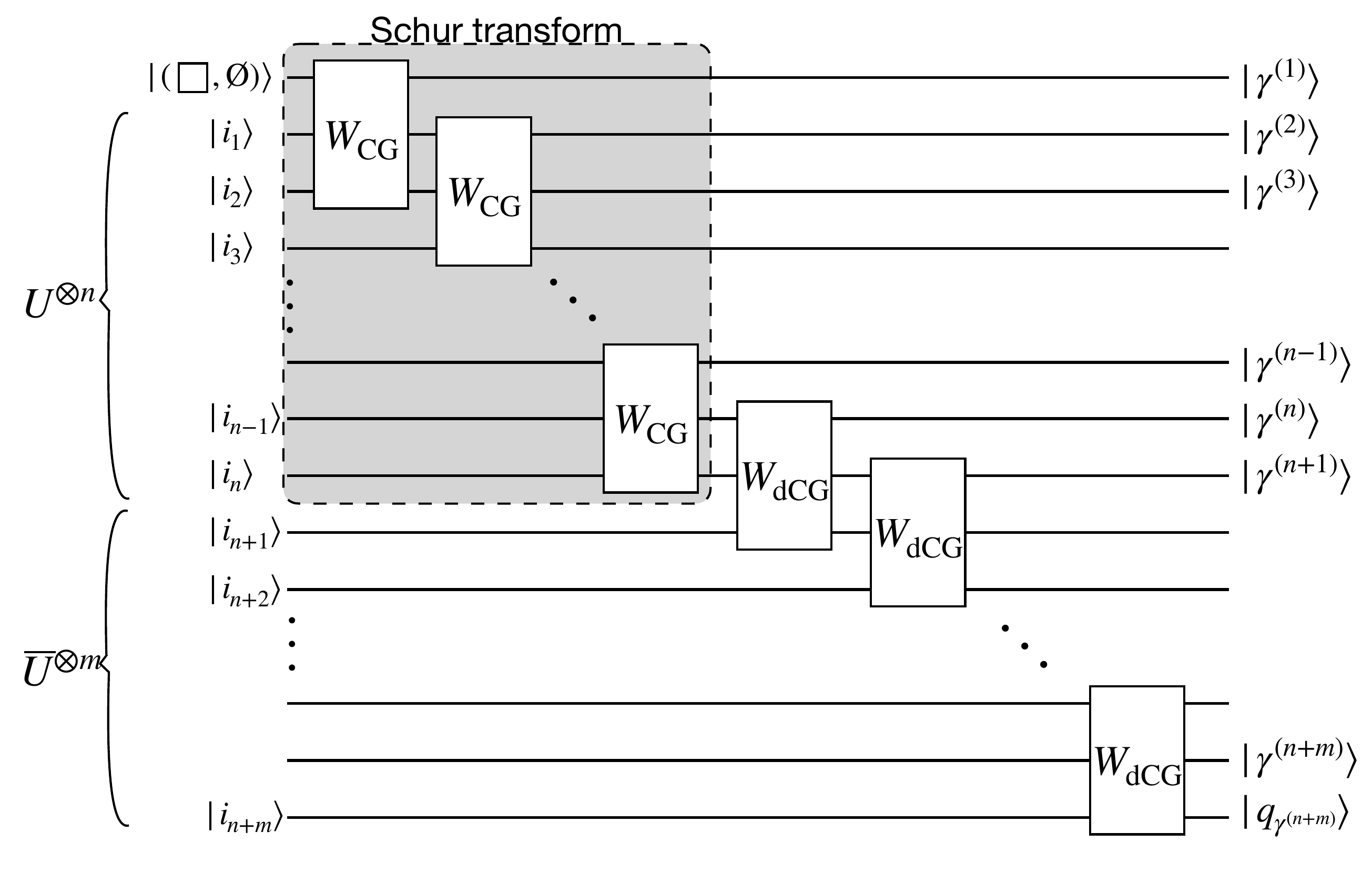}
    \caption{Clebsch-Gordan (CG) transforms $W_\mathrm{CG}$ are cascaded to produce the Schur transform (gray box)~\cite{bacon2007quantum}. Similarly, we construct a ``dual'' CG transform $W_\mathrm{dCG}$ for the dual defining representation $\bar{U}$ and obtain the mixed Schur transform. Here we perform all the CG transforms first and then the dual CG transforms, but the ordering of them can be arbitrary.}
    \label{fig:cascading}
\end{figure}

\subsection{Mixed Schur transform from CG and dual CG transforms}
Provided with efficient circuit constructions for the CG and dual CG transforms, we concatenate them to obtain the mixed Schur transform. Due to the mixed Schur-Weyl duality, the mixed Schur transform block-diagonalizes both the algebras $\mathcal{U}_{n,m}^d$ and $\mathcal{A}_{n,m}^d$.

Suppose we start with an input vector $\ket{i_1,\hdots,i_{n+m}}$, we combine the subsystems using
the CG transform, one at a time, up to subsystem $n$ as done in \cite{bacon2007quantum} (gray box in Fig.~\ref{fig:cascading}). We refer to~\cite{bacon2007quantum} for more details of this step.

After having performed the above normal Schur transform, the system is now of the form $\ket{\gamma^{(1)}}\hdots \ket{\gamma^{(n)}}\ket{q_{\gamma^{(n)}}} $, where $\gamma^{(k)}=[\alpha^{(k)},\varnothing]$. Continuing, we apply the dual CG transform on $\ket{[\alpha^{(n)},\varnothing]}\ket{q_{\gamma^{(n)}}} \ket{i_{n+1}}$ to produce $\ket{[\alpha^{(n)},\varnothing]} \ket{[\alpha^{(n+1)},\beta^{(n+1)}]}\ket{q_{\gamma^{(n+1)}}}$. An illustration of one step in this procedure is shown in the figure below.
\begin{figure}[H]
    \centering
    \includegraphics[width=0.5\textwidth]{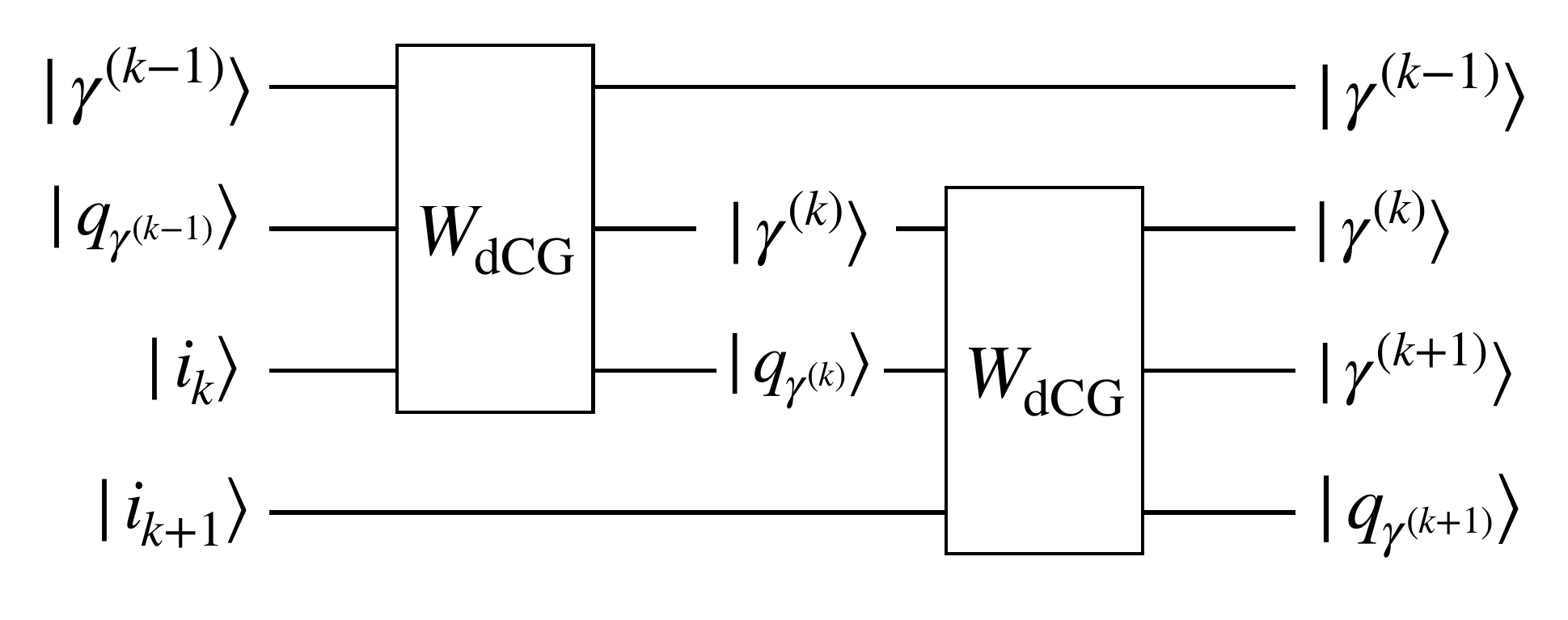}
    \caption{Concatenation of two dual CG transforms.}
    \label{fig:concatenateCG}
\end{figure}
Continuing adding qudits and applying the dual CG transform, we eventually obtain registers in the form of
\begin{equation}
    \ket{\gamma^{(1)}}\hdots \ket{\gamma^{(n+m)}}\ket{q_{\gamma^{(n+m)}}},
\end{equation}
but $\ket{\gamma^{(1)}}\hdots \ket{\gamma^{(n+m-1)}}$ describes exactly a path in the Bratteli diagram of $\mathcal{A}_{n,m}^d$ to $\gamma^{(n+m)}$, as described in Section \ref{sec:bratteli}! Hence, $\ket{\gamma^{(1)}}\hdots \ket{\gamma^{(n+m-1)}}$ represent a basis element $\ket{p_{\gamma^{(n+m)}}}$ of $\mathcal{P}^d_{\gamma^{(n+m)}}$, and we achieve registers of the form $\ket{\gamma,q_\gamma,p_\gamma}$in the RHS of Eq.~\eqref{eq:def-mixedschur} as desired.

The total complexity of the mixed Schur transform is $nT_\mathrm{CG}+ mT_{dCG}$, where $T_\mathrm{CG}=\mcO(d^3\operatorname{polylog}(n,d,1/\varepsilon))$ is the complexity of the CG transform derived in \cite{bacon2007quantum} and $T_\mathrm{dCG}=\mcO(d^4\operatorname{polylog}(n,m,d,1/\varepsilon))$ is the complexity of the dual CG transform described in Section~\ref{sec:recursive}.

\begin{remark} Above, we have presented the mixed Schur transform by first grouping irreps $U$ and invoking the Schur transform, and then adding $\Bar{U}$. However, they can in fact be combined in any arbitrary order as discussed in Remark~\ref{remark:insertionrule}. This is due to the flexible insertion rules of up-down staircase tableaux proved in~\cite{stembridge1987rational}.
\end{remark}

\section{Recursive circuit for the dual Clebsch-Gordan transform}\label{sec:recursive}
As discussed in the previous section, a naive implementation of the CG or dual CG transform would have a complexity of $(n+m)^{\mcO(d^2)}$. Bacon, Chuang, and Harrow~\cite{bacon2007quantum} improved the complexity of the CG transform to $\poly(d,\log n)$ by using a recursive construction based on the theory of tensor operators. Here we transfer their ideas to the case of dual CG transform, and in fact, more general classes of irreps, and further work out some details omitted in their paper.

We first define some convenient staircase notations that will be used in this section. Let $\varnothing \triangleq [\varnothing,\varnothing]$, $\square \triangleq [\square, \varnothing]$, and $\overline{\square} \triangleq [\varnothing, \square]$. For any $\mathbf{U}_d$ irrep $\gamma=[\alpha,\beta]$ and $j \in [d]$, denote by $\gamma - e_j$ the valid staircase resulting from, if $j \leq \operatorname{len}(\alpha)$, removing a box on row $j$ of $\alpha$, or, adding a box on row $j-\operatorname{len}(\alpha)$ of $\beta$. In other words, $\gamma - e_j \in \gamma - \square$ as defined in rule 2 in Section~\ref{sec:bratteli}. Denote by $\Gamma_d$ the set of all staircases of length $\leq d$.

\subsection{More representation theory facts}
For this section, we will additionally need the following facts from representation theory.

\begin{definition}[Isotypic decomposition, Proposition 4.1.15~\cite{goodman2009symmetry}] Let $V$ be a reducible representation of an algebra $\mcA$ (think group algera), such that $V \cong \bigoplus_{\lambda \in \hat{\mcA}} \otimes \mathbb{C}^{m_\lambda}$ where $\hat{\mcA}$ is the set of irreps and $m_\lambda$ is the multiplicty of $\lambda$. It holds that
\begin{equation*}
    \mathbb{C}^{m_\lambda} \cong \operatorname{Hom}_\mcA(V_\lambda, V),
\end{equation*}
where $\operatorname{Hom}_\mcA(V_\lambda, V)$ denotes the set of linear operators $O: V_\lambda \mapsto V$ that are invariant under the actions of $\mcA$ on $V_\lambda$ and $V$.
\label{def:isotypic}
\end{definition}

\begin{lemma}[Schur, Lemma 4.1.4~\cite{goodman2009symmetry}] Let $(r_1, V_1)$ and $(r_2,V_2)$ be irreps of an associative algebra $\mcA$. Assume that $V_1,V_2$ have countable dimensions over $\mcC$ Then
\begin{equation*}
    \operatorname{dim} \operatorname{Hom}_{\mathcal{A}}(V_1, V_2)=\left\{\begin{array}{l}
1 \text { if }(r_1, V_1) \cong(r_2, V_2) \\
0 \text { otherwise}
\end{array}\right. .
\end{equation*}
\label{lemma:Schur}
\end{lemma}

Consider any irreps $\mu,\lambda$ of an algebra $\mcA$. The space of linear operators $\operatorname{Hom}(V_\mu, V_\lambda)$ is also a representation since $\operatorname{Hom}(V_\mu, V_\lambda) \cong V_\mu ^* \otimes V_\lambda$ . The following fact follows from applying Definition~\ref{def:isotypic} to the irrep $\operatorname{Hom}(V_\mu, V_\lambda)$:
\begin{equation}
\begin{aligned}
\operatorname{Hom}(V_\mu, V_\lambda) & \cong \bigoplus_{\nu \in \hat{\mcA}} V_\nu \otimes \operatorname{Hom}_\mcA \left(V_\nu, \operatorname{Hom}\left(V_\mu, V_\lambda \right)\right) \\
& \cong \bigoplus_{\nu \in \hat{\mcA}} V_\nu \otimes \operatorname{Hom}_\mcA \left(V_\nu,V_\mu^* \otimes V_\lambda \right) \\
& \cong \bigoplus_{\nu \in \hat{\mcA}} V_\nu\otimes\left(V_\mu^* \otimes 
 V_\nu ^* \otimes V_\lambda \right)_\mcA \\
& \cong \bigoplus_{\nu \in \hat{\mcA}} V_\nu \otimes \operatorname{Hom}_\mcA \left(V_\mu \otimes V_\nu, V_\lambda \right).
\end{aligned}
\label{eq:irrep-multiplicity}
\end{equation}

\subsection{Tensor operators and the Wigner-Eckart theorem}
Let us start by reviewing the notion of irreducible tensor operators and the Wigner-Eckart theorem~\cite{biedenharn1968pattern, louck1970recent}, which relates the former to the generic CG transform. We will use this theorem crucially in the next subsection to construct the dual Clebsch-Gordan transform.

Let $\mathcal{Q}^d_\mu$ and $\mathcal{Q}^d_\lambda$ be irreps of $\mathbf{U}_d$ labelled by staircases $\mu, \lambda$. Observe that $\operatorname{Hom}(\mathcal{Q}^d_\mu, \mathcal{Q}^d_\lambda)$ is a representation of $\mathbf{U}_d$: for any $U \in \mathbf{U}_d$, $U$ acts on $\operatorname{Hom}(\mathcal{Q}^d_\mu, \mathcal{Q}^d_\lambda)$ by mapping $O \rightarrow \mathbf{q}^d_{\lambda}(U) O \mathbf{q}^d_{\mu}(U)^{-1}$ for $O \in \operatorname{Hom}(\mathcal{Q}^d_\mu, \mathcal{Q}^d_\lambda)$. Consider an irrep in this representation,
an \textit{irreducible tensor operator} is defined to be a collection of operators $T = \{T_1, T_2, \hdots\} \subset \operatorname{Hom}(\mathcal{Q}^d_\mu, \mathcal{Q}^d_\lambda)$ that forms a basis for an irrep of $\operatorname{Hom}(\mathcal{Q}^d_\mu, \mathcal{Q}^d_\lambda)$.
\begin{definition} Let $\mathcal{Q}^d_\bullet$, be irreps of $\mathbf{U}_d$ labelled by staircases $\bullet \in \{\mu, \lambda,\nu\}$.
    A set of operators $\mathbf{T}^{\nu}=\{T_{q_\nu}\}_{q_\nu \in \mathcal{Q}^d_\nu} \subset \operatorname{Hom}(\mathcal{Q}^d_\mu, \mathcal{Q}^d_\lambda)$, where $q_\nu$ label the GT basis elements of $\mathcal{Q}^d_\nu$, is an irreducible tensor operator corresponding to irrep $\nu$ if for any $U \in \mathbf{U}_d$ we have that
\begin{equation}
    \mathbf{q}^d_{\lambda}(U) T_{q_\nu} \mathbf{q}^d_{\mu}(U)^{-1}= \sum_{q_{\nu}' \in \mathcal{Q}^d_\nu} \bra{q_{\nu}'} \mathbf{q}^d_{\nu}(U) \ket{q_\nu} T_{q_{\nu}'}.
\end{equation}
\end{definition}

The Wigner-Eckart theorem provides a connection between the entries of an irreducible tensor operator and the CG transform.

\begin{theorem}[Wigner-Eckart, Theorem 7.2 in  \cite{harrow2005applications}]
Let $\mathbf{T}^{\nu}=\{T_{q_\nu}\}_{q_\nu \in \mathcal{Q}^d_\nu} \subset \operatorname{Hom}(\mathcal{Q}^d_\mu, \mathcal{Q}^d_\lambda)$ be an irreducible tensor operator.
Let $W_\mathrm{CG}^{\mu,\nu}$ denote the Clebsch-Gordan transform that decomposes $\mathcal{Q}^d_\mu \otimes \mathcal{Q}^d_\nu$ into irreps.
Furthermore, suppose that $\{\xi\}$ is an orthogonal basis of the multiplicity space $\operatorname{Hom}_{\mathbf{U}_d}(\mathcal{Q}^d_\mu \otimes \mathcal{Q}^d_\nu, \mathcal{Q}^d_\lambda)$ of $\mathcal{Q}^d_\lambda$ in $\mathcal{Q}^d_\mu \otimes \mathcal{Q}^d_\nu$. Then there exist scalars $\hat{T}_\xi^{\mu,\nu,\lambda}$ such that for any GT elements $\ket{q_\mu} \in \mathcal{Q}^d_{\mu}$, $\ket{q_\nu} \in \mathcal{Q}^d_{\nu}$, and $\ket{q_\lambda} \in \mathcal{Q}^d_{\lambda}$, it holds that
\begin{equation}
    \bra{q_\lambda} T_{q_\nu} \ket{q_\mu} = \sum_{\xi \in \operatorname{Hom}_{\mathbf{U}_d}(\mathcal{Q}^d_\mu \otimes \mathcal{Q}^d_\nu, 
    \mathcal{Q}^d_\lambda)} \hat{T}^{\mu,\nu,\lambda}_\xi \bra{\lambda, q_\lambda, \xi} W^{\mu,\nu}_\mathrm{CG} \ket{q_\mu,q_\nu}.
\end{equation}
\label{thm:wigner-eckart}
\end{theorem}

We refer to Chapter 18 of~\cite{vilenkin1992representations} for a mathematical proof and to Section 7 of~\cite{harrow2005applications} for another proof written in quantum information notations. It is worth noting that the theorem makes no assumptions about whether the irreps are polynomial or rational.

The Wigner-Eckart theorem relates the entries of a tensor operator to the entries of the CG transform (for input irreps $\mu, \nu$) via the so-called \emph{reduced Wigner coefficients} $\hat{T}^{\mu,\nu,\lambda}_\xi$ which, notably, are independent of the element $q_\nu$. Let us also note that this theorem is a generalization of the Schur's lemma, which corresponds to the special case when $\mathbf{T}^\nu$ transforms under the identity representation.

To apply this theorem to the dual CG transform, we consider the following specialized version when $\nu = \overline{\square}$. We allow $\mu$ to be arbitrary and thus insert the irrep label $\ket{\mu}$ as part of the input of the dual CG transform. Let $W_\mathrm{dCG}=\sum_{\mu} \ket{\mu}\bra{\mu}\otimes W_\mathrm{CG}^{\mu,\overline{\square}}$. In addition, we already see from the previous section that combining the irreps $\mu$ and $\overline{\square}$ only yields the irreps in $\mu - \square$. Furthermore, the output irreps are multiplicity-free, so we will not need the multiplicity label $\ket{\xi}$ in the output. We thus obtain the following specialized version of the Wigner-Eckart theorem.

\begin{theorem} Let $\gamma \in \mu - \square$ be a valid irrep appearing in the decomposition of $\mathcal{Q}^d_\mu \otimes \mathcal{Q}^d_{\overline{\square}}$. Let $\mathbf{T}^{\overline{\square}}=\{T_{i}\}_{i \in \mathcal{Q}^d_{\overline{\square}}} \subset \operatorname{Hom}(\mathcal{Q}^d_\mu, \mathcal{Q}^d_\gamma)$ be an irreducible tensor operator.
Then there exist scalars $\hat{T}^{\mu,\gamma}$ such that for any GT elements $\ket{q_\mu} \in \mathcal{Q}^d_{\mu}$, $\ket{i} \in \mathcal{Q}^d_{\overline{\square}}$, and $\ket{q_\gamma} \in \mathcal{Q}^d_{\gamma}$, it holds that
\begin{equation}
    \bra{q_\gamma} T_{i} \ket{q_\mu} = \hat{T}^{\mu,\gamma} \bra{\mu, \gamma, q_\gamma} W_\mathrm{dCG} \ket{\mu,q_\mu,i}.
\end{equation}
\label{thm:dualWE}    
\end{theorem}

\subsection{Dual CG transform using the Wigner-Eckart theorem}
We now apply the machinery in the previous subsection to construct an efficient quantum circuit for the dual CG transform. The strategy will be as follows. We first write the dual CG transform in terms of $\mathbf{U}_d$ tensor operators. Then we decompose these tensor operators into $\mathbf{U}_{d-1}$ tensor operators. Then we invoke the Wigner-Eckart theorem to relate these operators to the $\mathbf{U}_{d-1}$ dual CG transform. These steps yield an expression of the $\mathbf{U}_d$ dual CG transform in terms of the $\mathbf{U}_{d-1}$ dual CG transform and the reduced Wigner coefficients, which are efficiently computable. Recursively performing this procedure yields the desired dual CG transform. This decomposition is well-established in the representation theory of classical groups and was first used to construct the CG transform in~\cite{bacon2007quantum}. Here we show that it also works in the dual CG case where $\overline{\square}$ is an input. In fact, it works with a more general family of irreps. We provide more details in the parallel computation of the reduced Wigner coefficients that was omitted in~\cite{bacon2007quantum}.

Since we are recursively reducing the dimension of the unitary group, let us introduce a superscript used only in this subsection to label the dimension of the dual CG transform: a $\mathbf{U}_d$ dual CG transform (which performs the irrep decomposition of $\mathcal{Q}^d_\mu \otimes \mathcal{Q}^d_{\overline{\square}}$) is denoted as $W_\mathrm{dCG}^{[d]}$.

The first step in the recursive construction is to write the dual CG transform $W_\mathrm{dCG}^{[d]}$ in terms of certain $\mathbf{U}_d$ irreducible tensor operators.
\begin{prop}
\label{prop:dCG-tensor}
For any GT element $q_\mu \in \mathcal{Q}^d_\mu, i \in [d]$, and staircase $\mu$, it holds that
\begin{align}
    W_\mathrm{dCG}^{[d]} \ket{\mu}\ket{q_\mu}\ket{i} &= \ket{\mu} \sum_{j \in [d] \text{ s.t. } \atop \mu - e_j \in \Gamma_d }\sum_{q_{\mu-e_j} \in \mathcal{Q}^d_{\mu-e_j}} C^{\mu,j}_{q_\mu,i,q_{\mu-e_j}} \ket{\mu-e_j} \ket{q_{\mu-e_j}}\\
    &= \ket{\mu} \sum_{j \in [d] \text{ s.t. } \atop \mu - e_j \in \Gamma_d} \ket{\mu-e_j} T^{\mu,j}_i \ket{q_\mu}, \label{eq:def-T}
\end{align}
where we define the operators
\begin{equation}
    T_i^{\mu,j} \triangleq  \sum_{q_\mu \in \mathcal{Q}^d_\mu} \sum_{q_{\mu-e_j} \in \mathcal{Q}^d_{\mu-e_j}} C^{\mu,j}_{q_\mu,i,q_{\mu-e_j}}  \ket{q_{\mu-e_j}}\bra{q_\mu} \in \operatorname{Hom}(\mathcal{Q}^d_\mu, \mathcal{Q}^d_{\mu-e_j}).
\end{equation}
Furthermore, the set $T^{\mu,j} \triangleq \{T_{i}^{\mu,j}\}_{i \in [d]}$ is an irreducible tensor operator corresponding to irrep $\mathcal{Q}^d_{\overline{\square}}$.
\end{prop}
\begin{proof}
     By definition, the set $T^{\mu,j} \triangleq \{T_{i}^{\mu,j}\}_{i \in [d]}$ transforms according to $\mcQ^d_{\overline{\square}}$. Therefore, they form an irreducible tensor operator corresponding to the irrep $\mathcal{Q}^d_{\overline{\square}}$.
\end{proof}

The next step is to decompose $T^{\mu,j}$ into $\mathbf{U}_{d-1}$ irreducible tensor operators, which we will later express in terms of $W^{[d-1]}_\mathrm{dCG}$ using the Wigner-Eckart theorem.

\begin{lemma} For any irrep $\mu$, $i \in [d]$, and $j \in [d]$ such that $\mu - e_j$ is a $\mathbf{U}_d$ irrep, it holds that
    \begin{equation}
    T^{\mu,j}_i = \sum_{\mu' \precsim \mu } \sum_{j'=0}^{d-1} \ket{\mu' - e_{j'}}\bra{\mu'} \otimes T_i^{\mu,j,\mu',j'},
\end{equation}
where $\mu' \precsim \mu$ and $\mu' - e_{j'}$ are $\mathbf{U}_{d-1}$ irreps, and $T_i^{\mu,j,\mu',j'} \in \operatorname{Hom}(\mathcal{Q}^{d-1}_{\mu'}, \mathcal{Q}^{d-1}_{\mu'- e_{j'}})$.
\label{lem:tensor}
\end{lemma}
\begin{proof}
We first decompose each of $\mathcal{Q}^d_{\overline{\square}}$, $\mathcal{Q}^d_\mu$, and $\mathcal{Q}^d_{\mu-e_j}$ into $\mathbf{U}_{d-1}$ irreps.

Observe that $\mathcal{Q}^d_{\overline{\square}}$, as an irrep of $\mathbf{U}_{d-1}$, decomposes into $ \mathcal{Q}^{d-1}_{\overline{\square}} \oplus \mathcal{Q}^{d-1}_{\varnothing}$, where $\mathcal{Q}^{d-1}_{\overline{\square}}$ acts on the subspace $\operatorname{span}\{\ket{i} :i \in [d-1]\}$ and $\mathcal{Q}^{d-1}_{\varnothing}$ acts trivially on $\ket{d}$. Since we deal with unitary groups with different dimensions in this section, we will make it explicit under which group the irrep decomposition is being performed by stacking the group on top of the $\cong$ symbol
\begin{equation}  \mathcal{Q}^d_{\overline{\square}} \overset{\mathbf{U}_{d-1}}{\cong} \mathcal{Q}^{d-1}_{\overline{\square}} \oplus \mathcal{Q}^{d-1}_{\varnothing}.
\end{equation}
The previous decomposition implies that, for any $\mu,j$, we can identify the set  $\{T^{\mu,j}_i\}_{i\in [d-1]}$ as a $\mathbf{U}_{d-1}$ irreducible tensor operator corresponding to the dual irrep $\overline{\square}$ and the set $\{T^{\mu,j}_d\}$ (a set with a single operator) as another $\mathbf{U}_{d-1}$ irreducible tensor operator corresponding to the trivial irrep $\varnothing$.

Next, we decompose $\mathcal{Q}^d_\mu, \mathcal{Q}^d_{\mu-e_j}$ into $\mathbf{U}_{d-1}$ irreps. Recall from Section~\ref{sec:GTbasis} that $\mathcal{Q}^d_\mu \overset{\mathbf{U}_{d-1}}{\cong} \bigoplus_{\mu' \precsim \mu } \mathcal{Q}^{d-1}_{\mu'}$. Thus, we can decompose
\begin{equation}
    \operatorname{Hom}(\mathcal{Q}^d_\mu, \mathcal{Q}^d_{\mu-e_j}) \overset{\mathbf{U}_{d-1}}{\cong} \bigoplus_{\mu' \precsim \mu } \bigoplus_{\mu'' \precsim \mu - e_j} \operatorname{Hom}(\mathcal{Q}^{d-1}_{\mu'}, \mathcal{Q}^{d-1}_{\mu''}).
\end{equation}
This is where the convenience of the GT basis comes in. Recall from Section~\ref{sec:GTbasis} that in this basis, an element $\ket{q_\mu}$ in $\mathcal{Q}^d_\mu$ is of the form $\ket{\mu'}\ket{q_{\mu'}}$, where $\mu'$ interlaces $\mu$ and $\ket{q_{\mu'}}$ is an element in $\mathcal{Q}^{d-1}_{\mu'}$. Thus, we can decompose
\begin{equation}
    T^{\mu,j}_i = \sum_{\mu' \precsim \mu } \sum_{\mu'' \precsim \mu - e_j} \ket{\mu''}\bra{\mu'} \otimes T_i^{\mu,j,\mu',\mu''},
    \label{eq:decomposeT}
\end{equation}
where $T_i^{\mu,j,\mu',\mu''} \in \operatorname{Hom}(\mathcal{Q}^{d-1}_{\mu'}, \mathcal{Q}^{d-1}_{\mu''})$.

Now observe that many operators $T_i^{\mu,j,\mu',\mu''}$ in this sum are actually vanishing. Indeed, as $\{T^{\mu,j}_i\}_{i\in [d]}$  is a $\mathbf{U}_{d}$ irreducible tensor operator corresponding to $\mathcal{Q}^d_{\overline{\square}} \overset{\mathbf{U}_{d-1}}{\cong} \mathcal{Q}^{d-1}_{\overline{\square}} \oplus \mathcal{Q}^{d-1}_{\varnothing}$, the set $\{T^{\mu,j,\mu',\mu''}_i\}_{i \in [d-1]}$ must form a $\mathbf{U}_{d-1}$ irreducible tensor operator corresponding to the irrep $ \mathcal{Q}^{d-1}_{\overline{\square}}$, and similarly, the operator $\{T^{\mu,j,\mu',\mu''}_d\}$ corresponds to the irrep $\mathcal{Q}^{d-1}_{\varnothing}$. Thus, the only $\mathbf{U}_{d-1}$ irreps $\mu', \mu''$ that contribute to the sum in Eq. 
 \eqref{eq:decomposeT} are those such that $\operatorname{Hom}(\mathcal{Q}^{d-1}_{\mu'}, \mathcal{Q}^{d-1}_{\mu''})$ contains $\mathcal{Q}^{d-1}_{\overline{\square}}$ or $\mathcal{Q}^{d-1}_{\varnothing}$. In particular, when $i=d$, the operator $T^{\mu,j,\mu',\mu''}_d$ is a $\mathbf{U}_{d-1}$-invariant operator, i.e., $T^{\mu,j,\mu',\mu''}_d \in \operatorname{Hom}(\mathcal{Q}^{d-1}_{\mu'}, \mathcal{Q}^{d-1}_{\mu''})^{\mathbf{U}_{d-1}}$. According to Schur's lemma (Lemma~\ref{lemma:Schur}), $T^{\mu,j,\mu',\mu''}_d=0$ unless $\mu'=\mu''$, in which case $T^{\mu,j,\mu',\mu''}_d$ is a multiple of the identity operator, i.e., $T^{\mu,j,\mu',\mu'}_d\triangleq T^{\mu,j,\mu',0}_d = \hat{T}^{\mu,j,\mu',0} I_{\mathcal{Q}^{d-1}_{\mu'}}$ for some scalar $\hat{T}^{\mu,j,\mu',0}$ (the label $0$ is added for later notational convenience). In summary, when $i=d$, the dual CG transform acts as
 \begin{equation}
 \begin{aligned}
W_\mathrm{dCG}^{[d]}\ket{\mu}\ket{q_\mu} \ket{d} &\overset{\text{\eqref{eq:def-T}}}{=}  \sum_{j \in [d] \text{ s.t.} \atop \mu - e_j \in \Gamma_d } \ket{\mu} \ket{\mu-e_j} T_d^{\mu,j}\underbrace{\ket{q_\mu}}_{\text{equiv. to} \ket{\mu'}\ket{q_{\mu'}} \atop \text{ for some } \mu' \precsim \mu} \\ 
&= \sum_{j \in [d] \text{ s.t.} \atop \mu - e_j \in \Gamma_d } \ket{\mu} \ket{\mu-e_j} \ket{\mu'}\ket{q_{\mu'}} \hat{T}^{\mu,j,\mu',0}.
 \end{aligned}
 \label{eq:i=d}
 \end{equation}

The case $i \in [d-1]$ is slightly more involved. As $\{T^{\mu,j,\mu',\mu''}_i\}_{i \in [d-1]}$ transforms according to $\mathcal{Q}^{d-1}_{\overline{\square}}$, we need to determine whether $\operatorname{Hom}(\mathcal{Q}^{d-1}_{\mu'}, \mathcal{Q}^{d-1}_{\mu''})$ contains this irrep. Recall from Eq. \eqref{eq:irrep-multiplicity} that the multiplicity of $\mathcal{Q}^{d-1}_{\overline{\square}}$ is equal to the dimension of $\operatorname{Hom}_{\mathbf{U}_{d-1}}(\mathcal{Q}^{d-1}_{\mu'} \otimes \mathcal{Q}^{d-1}_{\overline{\square}}, \mathcal{Q}^{d-1}_{\mu''})$. We have that $\mathcal{Q}^{d-1}_{\mu'} \otimes \mathcal{Q}^{d-1}_{\overline{\square}} \overset{\mathbf{U}_{d-1}}{\cong} \oplus_{\mu'' \in \mu' -_{d-1} \square} \mathcal{Q}^{d-1}_{\mu''}$. Thus, by Schur's lemma, when $i \in [d-1]$, $T^{\mu,j,\mu',\mu''}_i=0$ unless $\mu'' = \mu' -e_{j'}$ for some $j'\in [d-1]$ such that $\mu' -e_{j'} \in \mu' -_{d-1} \square$.

To summarize, we have considered two cases, $i=d$ and $i \in [d-1]$, and decomposed $T^{\mu,j}_i$ into $\mathbf{U}_{d-1}$ irreducible tensor operators in each case.
We now define a unifying notation to combine these two cases. For $i \in [d]$, let $T_i^{\mu,j,\mu',j'}\triangleq T_i^{\mu,j,\mu',\mu' -e_{j'}}$, with the convention that $\mu'-e_0=\mu'$. In this notation, $T_d^{\mu,j,\mu',0}=T_d^{\mu,j,\mu',\mu'}= \hat{T}^{\mu,j,\mu',0} I_{\mathcal{Q}^{d-1}_{\mu'}}$, as defined in the previous paragraph, is the only nonzero operator when $i=d$, i.e., $T_d^{\mu,j,\mu',j'}=0$ for $j' \neq 0$.
Whereas for $i \in [d-1]$, $T_i^{\mu,j,\mu',j'}=0$ unless $\mu' -e_{j'} \in \mu' -_{d-1} \square$. So we can conveniently write
\begin{equation}
    T^{\mu,j}_i = \sum_{\mu' \precsim \mu } \sum_{j'=0}^{d-1} \ket{\mu' -e_{j'}}\bra{\mu'} \otimes T_i^{\mu,j,\mu',j'},
    \label{eq:tensor}
\end{equation}
where $T_i^{\mu,j,\mu',j'} \in \operatorname{Hom}(\mathcal{Q}^{d-1}_{\mu'}, \mathcal{Q}^{d-1}_{\mu' -e_{j'}})$.
\end{proof}

We are now ready to apply the Wigner-Eckart theorem on the $\mathbf{U}_{d-1}$ irreducible tensor operators $\{T_i^{\mu,j,\mu',j'}\}_{i \in [d-1]}$ in Lemma~\ref{lem:tensor} to relate them to the $\mathbf{U}_{d-1}$ dual CG transform $W^{[d-1]}_\mathrm{dCG}$. The irrep decomposition of $\mathcal{Q}^{d-1}_{\mu'} \otimes \mathcal{Q}^{d-1}_{\overline{\square}}$ into is multiplicity-free, so we can apply Theorem \ref{thm:dualWE}. In particular for any $\ket{q_{\mu'}} \in \mathcal{Q}^{d-1}_{\mu'}$ and $\ket{q_{\mu'-e_{j'}}} \in \mathcal{Q}^{d-1}_{\mu'-e_{j'}}$, where $j' \in [d-1]$ such that $\mu'-e_{j'} \in \mu' -_{d-1} \square$,
\begin{equation}
    \bra{q_{\mu'-e_{j'}}} T_i^{\mu,j,\mu',j'} \ket{q_{\mu'}} = \hat{T}^{\mu,j,\mu',j'} \bra{\mu',\mu'-e_{j'}, q_{\mu'-e_{j'}}} W^{[d-1]}_\mathrm{dCG} \ket{\mu',q_{\mu'},i},
    \label{eq:d-1Wigner}
\end{equation}
where $\hat{T}^{\mu,j,\mu',j'}$ are the corresponding reduced Wigner coefficients.

For later convenience, we collect Equations~\eqref{eq:i=d},~\eqref{eq:tensor},~\eqref{eq:d-1Wigner} and Lemma~\ref{lem:tensor} into the following proposition that relates between $T_i^{\mu,j}$ and $W^{[d-1]}_\mathrm{dCG}$. 

\begin{prop} Consider any irrep $\mu$, and $j\in [d]$, such that $\mu -e_j$ is a $\mathbf{U}_d$ irrep. For $i=d$, it holds that
\begin{equation*}
    T^{\mu,j}_i = \sum_{\mu' \precsim \mu} \sum_{q_{\mu'} \in  \mathcal{Q}^{d-1}_{\mu'} } \ket{\mu', q_{\mu'}} \bra{\mu', q_{\mu'}} \cdot \hat{T}^{\mu,j,\mu',0},
\end{equation*}
and for $i\in [d-1]$, it holds that
\begin{equation*}
     T^{\mu,j}_i  = \sum_{\mu' \precsim \mu} \sum_{j'=0}^{d-1} \sum_{q_{\mu'} \in \mathcal{Q}^{d-1}_{\mu'} \atop  q_{\mu'-e_{j'}}  \in \mathcal{Q}^{d-1}_{\mu'-e_{j'}}} \ket{\mu'-e_{j'}, q_{\mu'-e_{j'}} } \bra{\mu', q_{\mu'} } \cdot \hat{T}^{\mu,j,\mu',j'} \bra{\mu',\mu'-e_{j'}, q_{\mu'-e_{j'}}} W^{[d-1]}_\mathrm{dCG} \ket{\mu',q_{\mu'},i},
\end{equation*}
where $\hat{T}^{\mu,j,\mu',j'}$ are the reduced Wigner coefficients.
\label{prop:recursive}
\end{prop}

Finally, we combine Propositions~\ref{prop:dCG-tensor} and \ref{prop:recursive} as follows. When $i \in [d-1]$, the dual CG transform acts as
\begin{equation}
\begin{aligned}
W_\mathrm{dCG}^{[d]}\ket{\mu}\ket{q_\mu} \ket{i}&\overset{\text{\eqref{eq:def-T}}}{=}  \sum_{j \in [d] \text{ s.t.} \atop \mu -e_j \in \Gamma_d } \ket{\mu} \ket{\mu-e_j} T_i^{\mu,j}\underbrace{\ket{q_\mu}}_{\text{equiv. to} \ket{\mu'}\ket{q_{\mu'}} \atop \text{ for some } \mu' \precsim \mu} \\ 
&\overset{\text{\eqref{eq:tensor}}}{=} \sum_{j\in[d] \text{ s.t.} \atop \mu -e_j \in \Gamma_d} \sum_{j'\in [d-1] \text{ s.t.} \atop \mu' -e_{j'} \in \Gamma_{d-1}} \ket{\mu} \ket{\mu-e_j}  \ket{\mu' -e_{j'}}  T_i^{\mu,j,\mu',j'} \ket{q_{\mu'}}\\
&\overset{\text{\eqref{eq:d-1Wigner}}}{=} \sum_{j,j', q}  \ket{\mu} \ket{\mu-e_j}  \ket{\mu' -e_{j'}} \ket{q} \hat{T}^{\mu,j,\mu',j'} \bra{\mu',\mu'-e_{j'}, q} W^{[d-1]}_\mathrm{dCG} \ket{\mu',q_{\mu'},i},
\end{aligned}
\label{eq:i<d}
\end{equation}
where we denoted $q\equiv q_{\mu'-e_{j'}} \in \mathcal{Q}^{d-1}_{\mu'-e_{j'}}$ and omitted the conditions on $j,j'$ in the last sum for brevity. Observe that $\ket{\mu'-e_{j'}}\ket{q_{\mu'-e_{j'}}}$ is nothing but an element in $\mathcal{Q}^d_{\mu-e_j}$ according to the GT basis! Of course, one might concern that $\mu'-e_{j'}$ does not interlace $\mu-e_j$, but the reduced Wigner coefficients $\hat{T}^{\mu,j,\mu',j'}$ conveniently vanish in this case as we will see next subsection.

We can incorporate the case $i=d$ of Eq. \eqref{eq:i=d} into \eqref{eq:i<d} by replacing $W^{[d-1]}_\mathrm{dCG}$ with
\begin{equation}
    \widetilde{W}^{[d-1]}_\mathrm{dCG} = \left(\sum_{\mu', q_{\mu'} \in \mathcal{Q}^{d-1}_{\mu'}}  \ket{\mu',\mu',q_{\mu'}}\bra{\mu', q_{\mu'}, d}\right) +W^{[d-1]}_\mathrm{dCG}.
    \label{eq:modified-WdCG}
\end{equation}
\begin{figure}
    \centering
    \includegraphics[width=0.7\textwidth]{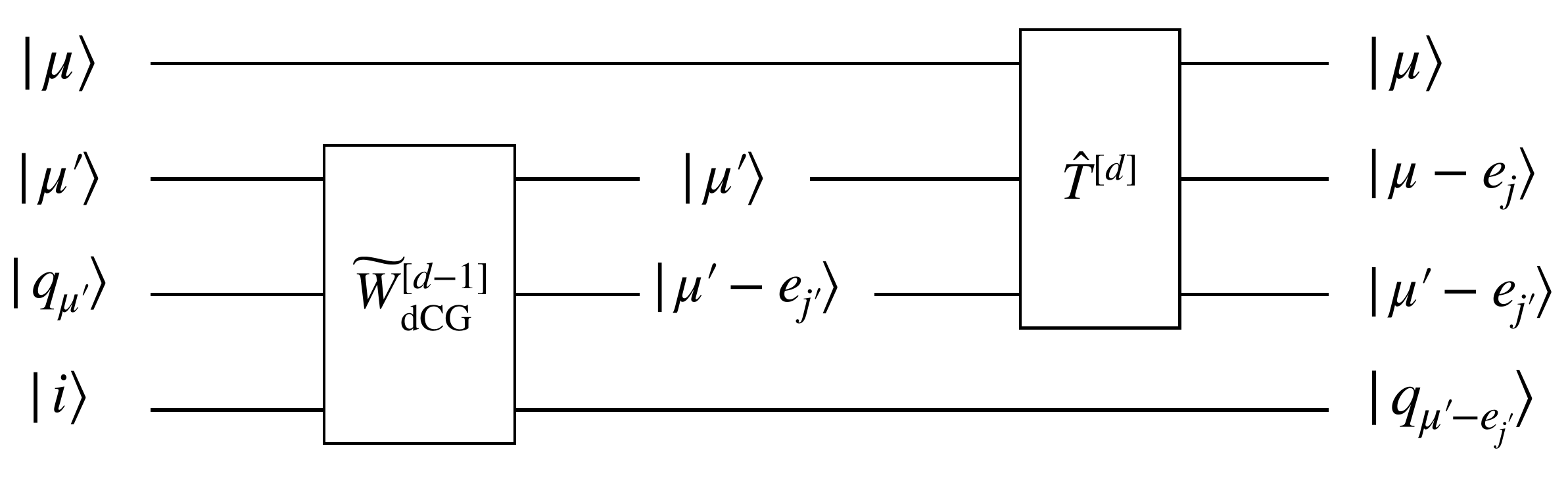}
    \caption{The $\mathbf{U}_d$ dual CG transform is decomposed into the $\mathbf{U}_{d-1}$ dual CG transform (in particular the modified version in Eq.~\eqref{eq:modified-WdCG}) and the reduced Wigner operator $T^{[d]}$.}
    \label{fig:recursive}
\end{figure}
With this, we are ready to fully describe the recursive step in the implementation of $W^{[d]}_\mathrm{dCG}$. For an input $\ket{\mu}\ket{q_\mu} \ket{i} = \ket{\mu} \ket{\mu'} \ket{q_{\mu'}} \ket{i}$, we first apply $\widetilde{W}^{[d-1]}_\mathrm{dCG}$ on $\ket{\mu'} \ket{q_{\mu'}} \ket{i}$ to get a superposition of registers of the form $\ket{\mu'}\ket{\mu'-e_{j'}}\ket{q_{\mu'-e_{j'}}}$. Then we apply the operator 
\begin{equation}
    \hat{T}^{[d]} \triangleq \sum_{\mu} \sum_{ \mu'\precsim \mu} \sum_{j=1}^d \sum_{j'=0}^{d-1} \hat{T}^{\mu,j,\mu',j'} \ket{\mu, \mu-e_j,\mu'-e_{j'}}\bra{\mu,\mu',\mu'-e_{j'}} 
\end{equation} on the registers $\ket{\mu}\ket{\mu'}\ket{\mu'-e_{j'}}$. See Figure~\ref{fig:recursive}. We continue the recursion until $d=1$, when the dual CG transform can be trivially implemented. Note that $\hat{T}^{[d]}$ is a unitary operator because $W^{[d]}_\mathrm{dCG}$ and $\widetilde{W}^{[d-1]}_\mathrm{dCG}$ are by definition. The next subsection shows how to implement $\hat{T}^{[d]}$ efficiently.

\subsection{Implementing the reduced Wigner operator}

Observe that $\hat{T}^{[d]}$ is effectively a controlled $d\times d$ operator. Indeed, it performs a unitary operator (from $0\leq j' \leq d-1$ to $1\leq j \leq d$) conditioned on $\mu,\mu'$; and the entries of this rotation are the reduced Wigner coefficients $\hat{T}^{\mu,j,\mu',j'}$ in Eq.~\eqref{eq:d-1Wigner}, whose efficiently computable formula was first derived by Biedenharn and Louck~\cite{biedenharn1968pattern}. We use the following simplified formula according to Equations 5 and 6, Section 18.2.10 in the textbook by Vilenkin and Klimyk~\cite{vilenkin1992representations}.

\begin{theorem}[Reduced Wigner coefficients for the dual irrep~\cite{vilenkin1992representations}] Define $s_{i}=\mu_i - i$ for $i\in [d]$ and $s'_{i}=\mu'_i - i$ for $i \in [d-1]$, and $S(j,j')=1$ if $j \leq j'$ and $S(j,j')=-1$ if $j > j'$ . Then, for $j'=0$,
\begin{equation}
    \hat{T}^{\mu,j,\mu',0}= \left| \frac{\displaystyle \prod_{k=1}^{d-1} (\mu'_k - \mu_k) }{\displaystyle \prod_{k \neq i}^{d} (\mu_k -\mu_j) }  \right|^{1/2},
\end{equation}
and for $j' \in [d-1]$,
\begin{equation}
\begin{aligned}
        &\hat{T}^{\mu,j,\mu',j'}= S(j,j')  \left|\left(\prod_{k=1}^{d-1} 
        \frac{\mu'_k - \mu_k}{\mu'_k -\mu'_j +1} \right) \left( \prod_{k=1}^{d}
        \frac{\mu_k-\mu'_j +1}{\mu_k -\mu_j} \right) \right|^{\frac{1}{2}}.
\end{aligned}
\end{equation}
\end{theorem}
Conditioned on $\mu, \mu'$, each coefficient $\hat{T}^{\mu,j,\mu',j'}$ can be computed in time $d\operatorname{polylog}(d,n,m, 1/\varepsilon)$ to any desired error $\varepsilon$ and will be computed in parallel conditioned on the registers $\ket{\mu,\mu',\mu'-e_j'}$ as described below.

\begin{remark} The reduced Wigner operator $\hat{T}^{[d]}$ can be implemented with a circuit complexity of $\mcO(d^3 \polylog(d,n,m,1/\varepsilon)$ via the following procdure:
\begin{enumerate}
    \item \textbf{Input}: register of the form $\ket{\mu,\mu',\mu'-e_j'}$ (see Fig.~\ref{fig:recursive}).
    \item Conditioned on $\ket{\mu,\mu'}$, compute the entries $\hat{T}^{\mu,j,\mu',j'}$ for all $0 \leq j' \leq d-1$ and $j \in [d]$ and write all of them on a working space. This is a purely classical and takes time $\Tilde{\mcO}(d^3)$.
    \begin{equation}
        \ket{\mu,\mu',\mu'-e_j'} \mapsto \ket{\mu,\mu',\mu'-e_j'} \otimes \ket{\emph{all matrix entries}}
    \end{equation}
    \item Use the procedure in Section 4.5 of Nielsen and Chuang~\cite{nielsen2000quantum} to decompose this $d \times d$ unitary into $T=\mcO(d^2 \polylog(d,n,m,1/\varepsilon))$ CNOT, Hadamard, and $Z^{1/4}$ gates acting on $\log_2 d$ qubits,
    \begin{equation}
            \ket{\mu,\mu',\mu'-e_j'} \otimes \ket{\emph{all matrix entries}} \mapsto \ket{\mu,\mu',\mu'-e_j'} \otimes \ket{\emph{gate decomposition}},
    \end{equation}
    where $\ket{\emph{gate decomposition}} \equiv \bigotimes_{t=1}^{T} \ket{\emph{gate}_t}$ and the qubits that $\emph{gate}_t$ acts on are fixed independent of $\mu,\mu',j,j'$ via some canonical choice of gate and location ordering. This takes time $\Tilde{\mcO}(d^2)$.
    \item First classically perform $\ket{\mu,\mu',\mu'-e_j'} \mapsto  \ket{\mu,j',\mu'-e_j'} $. Then
    for each subregister $\ket{\emph{gate}_t}$, perform the described gate (either $CNOT, H$, $Z^{1/4}$, or $I$) on the corresponding location in the subregister holding $\ket{j'}$. This latter step takes time $\Tilde{\mcO}(d^2)$
    \begin{equation}
        \ket{\mu, j',\mu'-e_j'} \mapsto \sum_{j} \hat{T}^{\mu,j,\mu',j'}  \ket{\mu, j,\mu'-e_j'}
    \end{equation}
    \item Undo the classical computation in steps 2,3 and discard the working space qubits in time $\Tilde{\mcO}(d^3)$.
    \item Classically perform $\ket{\mu, j,\mu'-e_j'} \mapsto \ket{\mu, \mu-e_j,\mu'-e_j'}$.
\end{enumerate}
\end{remark}

All in all, we obtain the following circuit complexities for the dual CG transform and mixed Schur transform.
\begin{cor} The dual Clebsch-Gordan transform can be implemented to precision $\varepsilon$ using $\mcO(d^4 \polylog(d,n,m,1/\varepsilon)$ elementary gates and $\mcO(d^2 \polylog(d,n,m,1/\varepsilon)$ ancilla qubits that are discarded at the end.
\end{cor}

\begin{cor}
The mixed Schur transform can be implemented to precision $\varepsilon$ using $\widetilde{\mcO}((n+m)d^4)$ elementary gates, $\widetilde{\mcO}(n+m+d^2)$ ancilla qubits (that are needed to label the irreps and GT basis vectors), and additionally $\widetilde{\mcO}(d^2)$ workspace ancilla qubits that are discarded at the end.
\label{cor:main}
\end{cor}
\begin{proof}
    As shown in Section~\ref{sec:transform}, the total complexity of the mixed Schur transform is $nT_\mathrm{CG}+ mT_{dCG}$, where $T_\mathrm{CG}=\mcO(d^3\operatorname{polylog}(n,d,1/\varepsilon))$ is the complexity of the CG transform derived in \cite{bacon2007quantum} and $T_\mathrm{dCG}=\mcO(d^4\operatorname{polylog}(n,m,d,1/\varepsilon))$ is the complexity of the dual CG transform. We need $\mcO(\min\{(n+m)\log(n+m), (m+m)\log d \})$ qubits to label the staircases $\gamma$, cf. Remark~\ref{remark:countirreps}, $\mcO(d^2\log(m+n))$ qubits to encode a Gelfand pattern for $\ket{q_\gamma}$, cf. Section~\ref{sec:GTbasis}, and $\mcO((n+m)\log d)$ qubits to encode a path in the Bratteli diagram for $\ket{p_\gamma}$ (cf. Section~\ref{sec:bratteli}). The extra $\Tilde{\mcO}(d^2)$ workspace qubits used for computing the reduced Wigner coefficients are reset after every dual CG transforms.
\end{proof}

\begin{remark} Finally, we remark that extensions of the CG or dual CG transforms to other input irreps are in fact possible using the same tensor operator framework. Indeed, its key idea--recursively reducing the $d$-dimensional to $(d-1)$-dimensional CG transforms via the reduced Wigner coefficients--is available as long as the reduced Wigner coefficients can be efficiently computed. This is the case when the input irrep $\nu$ is so-called elementary irreps, which are of the form $\nu=(c^{k},(c-1)^{d-k})$ for any integer $c$ (see Section 18.2.9~\cite{vilenkin1992representations}). It's also the case for irreps of the form $\nu=( c+s, c^{d-1})$ or $\nu=(c^{d-1},c-s)$ for positive integer $s$~\cite[Section 18.2.6]{vilenkin1992representations}.
\end{remark}

\section{Applications}\label{sec:applications}
In this section, we provide some applications of the mixed Schur transform. While we place our main focus on the implementation of unitary-equivariant channels, we believe there are more applications of it to be found in quantum algorithms, quantum learning and property testing, and complexity theory, which we briefly discuss and justify.

\subsection{Implementing unitary-equivariant channels}\label{sec:implementation}
Grinko and Ozols~\cite{grinko2022linear} recently provides a classical algorithm for solving linear programs with unitary-equivariant constraints, with possible extensions to semidefinite programs. In particular, they consider various optimization problems of unitary-equivariant channels and develop a diagrammatic framework for solving them using the mixed Schur-Weyl duality. These channels have a wide range of applications in quantum information in tasks such as quantum majority vote \cite{buhrman2022quantum}, asymmetric cloning \cite{nechita2021geometrical}, multiport-based teleportation \cite{kopszak2021multiport}, black-box unitary transformation~\cite{yoshida2022reversing}, unitary-equivariant variational ansatzes~\cite{nguyen2022theory}, etc. Others include situations that involve the partial transpose on subsystems, such as entanglement detection~\cite{bardet2020characterization, huber2022dimension}. Despite their wide applications, efficient implementations of these channels have been, to our knowledge, underexplored. A prior work~\cite{buhrman2022quantum} resolves this question for the special case of $d=2$, $m$ odd, and $n=1$ in the context of quantum majority vote. However, using counting argument, the exponentially large dimension of $\mathcal{A}_{n,m}^d$ will exclude any \textit{general} efficient implementation.

For the above reason, one must consider this task case by case or under reasonable input models. In this section, we discuss how to use the mixed Schur transform to achieve this task in various settings, partially answering an open question in~\cite{grinko2022linear} where the authors asked how to implement these channels provided with an efficient description of the Choi state in the Schur basis.

As a reminder, an $m$ to $n$ qudit channel $\mcN$ is said to be unitary-equivariant if
\begin{equation}
    \mathcal{N}(U^{\otimes m} \rho U^{\dagger \otimes m}) = U^{\otimes n}  \mathcal{N}(\rho )U^{\dagger \otimes n}, \qquad \forall U \in \mathbf{U}_d.
\end{equation}
Consider the Choi state $J^{\mcN} \triangleq (\id_{d^m} \otimes \mcN) (\ket{\Omega_d}\bra{\Omega_d}^{\otimes m})$, where $\ket{\Omega_d}$ is the normalized $d$-dimensional EPR state, unitary equivariance can be shown to be equivalent to~\cite{wolf2012quantum}
\begin{equation}
    [J^\mathcal{N}, \Bar{U}^{\otimes m} \otimes U^{\otimes n}]=0, \qquad \forall U \in \mathbf{U}_d.
\end{equation}
Later, we also consider the unnormalized Choi state $\Tilde{J}^{\mcN} = d \cdot J^\mcN$. Note the following identity~\cite{wolf2012quantum}
\begin{equation}
    \mcN(\rho) = \Tr_A [\Tilde{J}_{AB}^\mcN ( \rho^\top_A  \otimes \id_B) ].
\end{equation}

\subsubsection{Teleportation-based implementation}
Here we consider an implementation of unitary-equivariant channels based on (equivariant) teleportation. The main idea is to employ the standard gate teleportation scheme and exploiting the symmetry properties of the channel~\cite[Section 2]{wolf2012quantum}. This method works well when the input size is $m=1$ or small constant.

\begin{assumption} The entries of the Choi state in the Schur basis are known and efficiently computable, and the Schur basis Choi state $J^\mcN_\mathrm{Sch}$ can be prepared as a quantum state at unit cost.
\label{assumption}
\end{assumption}

We start with the case $m=1$ and show the channel can be implemented deterministically from the Choi state.

\begin{prop}[Teleportation-based implementation] In the case $m=1$, suppose the Choi state $J^\mcN$ of a qudit unitary-equivariant channel $\mcN$ is given in the Schur basis as a quantum state, then $\mcN$ can be implemented deterministically using one copy of $J^\mcN$ with a circuit complexity of $\mcO(d^6 + n)$.
\label{prop:cov-tele}
\end{prop}
\begin{proof} 
    Let $\rho \in \mathbb{C}^d \times \mathbb{C}^d $ be the input state. First,
     we apply the mixed Schur transform to convert the Choi state $J^\mcN_{AB}$ to the computational basis, which costs $\widetilde{\mcO}(n d^4)$ gate complexity. Consider a teleportation protocol between two parties: Alice and Bob, who hold each half $A$ and $B$ of $J^\mcN$, respectively.   
     Alice performs the following joint POVM on her half of the Choi state $J^\mathcal{N}$ and $\rho$:
    \begin{equation}
        \{M^U\}_{U \in \mathbf{U}_d} =\{d^2 \cdot (\id_{d} \otimes \bar{U}) \ket{\Phi_{d}}\bra{\Phi_d} (\id_{d} \otimes U^\top) d\mu(U) \}_{U \in \mathbf{U}_d},
        \label{eq:cov-povm}
    \end{equation}
    Note that this is a continuous POVM and will need to be discretized in an efficiently implementable way that we describe later. Let us first verify this is a valid POVM. Indeed, by construction, $M^U \geq 0$ for any $U \in \mathbf{U}_d$ and 
    \begin{equation}
        \int_{U \in \mathbf{U}_d} d^2 \cdot (\id_{d} \otimes \bar{U}) \ket{\Phi_{d}}\bra{\Phi_d} (\id_{d} \otimes U^\top) d\mu(U) = \id_{d^2},
    \end{equation}
    where we have used the identity $\int_{U \in \mathbf{U}_d} \Bar{U}_B O_{AB} U_B^\top = \frac{\Tr_B[O]}{d} \otimes \id_B$ which in turn follows from $\Bar{U}$ being irreducible and the Schur's lemma.
    
    Upon Alice measuring an element $U$ in this set, Bob performs the correction $U^{\dagger \otimes n}$ on his half of the Choi state $J^\mcN$ to get the output $\mathcal{N}(\rho)$. The entire protocol is equivalent to the following channel
    \begin{equation}
        \mcC(\rho_A \otimes J^\mcN_{A'
        B}) = \int_{U \in \mathbf{U}_d} d\mu(U) U^{\dagger \otimes n} \Tr_{AA'} [(\rho_A \otimes J^\mcN_{A'B}) (M^U_{AA'} \otimes \id_B)] U^{\otimes n},
    \end{equation}
    which we now verify is equal to $\mcN(\rho)$. Indeed, we have that
    \begin{align*}
        &\Tr_{AA'} [(\rho_A \otimes J_{A'B}^\mcN) (M^U_{AA'} \otimes \id_B)]\\
        & = d^2 d\mu(U) \Tr_{AA'} [(\rho_A \otimes  (\id_{A'}\otimes \mcN_{A'' \mapsto B}) (\ket{\Phi_d}\bra{\Phi_d}_{A'A''})) ((\id_A \otimes \bar{U}_{A'}) \ket{\Phi_{d}}\bra{\Phi_d}_{AA'} (\id_A \otimes U^\top_{A'}) \otimes \id_B)]\\
        & = d^2 d\mu(U)\Tr_{AA'} [(\rho_A \otimes  (\id_{A'}\otimes \mcN_{A'' \mapsto B}) (\id_A \otimes U_{A''} \ket{\Phi_d} \bra{\Phi_d}_{A'A''} \id_A \otimes U^\dagger_{A''} ) ) (\ket{\Phi_{d}}\bra{\Phi_d}_{AA'}  \otimes \id_B)]\\
        & = d^2 d\mu(U) U^{\otimes n}_{B}  \Tr_{AA'} [(\rho_A \otimes  (\id_{A'}\otimes \mcN_{A'' \mapsto B}) (\ket{\Phi_d} \bra{\Phi_d}_{A'A''}) \ket{\Phi_{d}}\bra{\Phi_d}_{AA'}] U^{\dagger \otimes n}_{B}\\
        & = d\mu(U) U^{\otimes n}  \mcN (\rho) U^{\dagger \otimes n}.
    \end{align*}
    Above, the first equality plugs in definitions, the second and fourth use the transpose trick with the EPR state, and the third equality follows from unitary equivariance.

    Finally we replace the continuous POVM over the unitary group by an efficient discrete POVM which effects it. In this case, it suffices to replace the Haar measure in Eq.~\eqref{eq:cov-povm} by a 1-design. For this, we simply use the uniform distribution over the Weyl operators~\cite{watrous2018thetheory}
    \begin{align}
        W_{a,b} = T^a P^b, &\qquad (a,b) \in \mathbb{Z}_d \times \mathbb{Z}_d,\\
        \text{where }&T = \sum_{j} \ket{j+1}\bra{j} \qquad \text{and} \qquad P = \sum_{j} e^{2\pi j /d} \ket{j} \bra{j}.
    \end{align}
    The discretized POVM can then be implemented by the isometry
    \begin{equation}
        W = \frac{1}{d} \sum_{a,b \in \mathbb{Z}_d \times \mathbb{Z}_d} \sqrt{ (\id_{d} \otimes \bar{W}_{a,b}) \ket{\Phi_{d}}\bra{\Phi_d} (\id_{d} \otimes W_{a,b}^\top)} \otimes \ket{a,b},
    \end{equation}
    which costs $O(d^6)$ elementary gates~\cite[Theorem 1]{iten2016quantum}. Upon receiving the measurement outcome $(a,b)$, Bob applies to correction $W_{a,b}^{\otimes n}$, giving the total gate complexity as specified.
\end{proof}

The dependence $d^6$ is somewhat unpleasing, but we do not attempt to optimize it here. One immediate improvement can perhaps be achieved by reducing the number of POVM elements with approximate designs from random circuits~\cite{brandao2016local, haferkamp2022random}.

\paragraph{Applications with $m=1$} Some prototypical applications with $m=1$ includes covariant error correction~\cite{kong2022near, faist2020continuous}, quantum state cloning and gate super-replication~\cite{fan2014quantum, nechita2021geometrical, chiribella2015universal}, variational quantum algorithms~\cite{nguyen2022theory, zheng2022super}. In particular, an encoding map in a covariant quantum error-correcting code with transversal gates is essentially a covariant channel under the relevant group, in the sense that if the logical qudit is applied $U$, its encoding is applied $U^{\otimes n}$. While a universal set of gates cannot be implemented transversally, they can for approximate codes. For (asymmetric) quantum state cloning, the standard setting is where one is provided with a qudit state $\rho$. For gate super-replication, one is given a small uses of an unknown unitary $U$ and the task is to generate more uses of that unitary. In variational quantum machine learning, it is recently introduced in~\cite{nguyen2022theory} `lifting' equivariant layers, which increase the number of qubits and thus have effects opposite to so-called pooling layers. Ref.~\cite{nguyen2022theory} derived a complete set of $1$-to-$2$ qubit unitary-equivariant channels, which we also analyzed as an example in Section~\ref{sec:example}.

\paragraph{Random unitary-equivariant channels} Recent work \cite{marvian2022restrictions} show that universal symmetric unitaries cannot be generated from local symmetric unitaries, hence it is non-trivial to generate $k$-design symmetric unitaries. We observe that the quantum Schur transform provides a way to overcome this issue using controlled local random unitaries. In particular, we can efficiently generate $k$-design unitaries \cite{brandao2016local, haferkamp2022random} using local gates (controlled on the irrep label register) within each block of the Schur basis, where there are no symmetry restrictions. Then, applying the Schur transform yields a $k$-design $\mathbf{U}_d$-symmetric unitary. This has applications in covariant quantum error correction~\cite{kong2022near} where one requires an efficiently generated random symmetric unitary to construct random equivariant codes. Similarly, the mixed Schur transform also provides a way to generate random equivariant channels by first generating a random Choi state in the Schur basis, applying the mixed Schur transform, and then teleporting as described above.

\paragraph{Generalization to $m=O(1)$} Finally, we note the teleportation protocol in Proposition~\ref{prop:cov-tele} can be straightforwardly generalized to $m>1$ at the cost of losing determinism. The reason is because the $m>1$ version of the ensemble $\{M^U\}$ in Eq.~\eqref{eq:cov-povm} no longer constitutes a complete POVM. However, the method still gives success rate independent of $d$ when $m$ is small. Indeed, we  can now replace the POVM element in Eq.~\eqref{eq:cov-povm} by 
    \begin{equation}
        M = C \int_{U \in \mathbf{U}_d} d\mu(U)  M^U,
    \end{equation}
    where 
    \begin{equation}
        M^U = (\id_{d^m} \otimes \bar{U}^{\otimes m}) \ket{\Phi_{d}}\bra{\Phi_d}^{\otimes m} (\id_{d^m} \otimes U^{\top \otimes m}),
    \end{equation}
and let $M^\perp = \id - M$ and $C$ is a normalization constant to ensure $\|M\| \leq 1$ (e.g. we chose $C=d^2$ when $m=1$, which led to $M=\id$). As done before, this POVM can be discretized by using (approximate) unitary $m$-designs.

Next, we calculate in terms of $C$ the probability of measuring $M$ on the joint system of $\rho$ and Alice's half of $J^\mcN$ (in which case Bob can apply the correction and the channel is successfully implemented). Since to Alice's perspective, her half of the Choi state is the maximally mixed state $\frac{\id_d^{\otimes m}}{d^m}$, the probability of measuring $M$ is
\begin{equation}
    \Prob[M] = \Tr M  \left( \rho \otimes \frac{\id_d^{\otimes m}}{d^m} \right) = \frac{C}{d^m} \Tr[\ket{\Phi_{d}}\bra{\Phi_d}^{\otimes m} (\rho \otimes \id_d^{\otimes m})] = \frac{C}{d^{2m}}.
\end{equation}

Thus, we would like to choose $C$ as large as possible while keeping $\|M\| \leq 1$. To upper bound $C$, we use the following formula from Weingarten calculus~\cite{collins2006integration}
\begin{equation}
    \int_{\mathbf{U}_d} d\mu(V) V^{\otimes m} B V^{\dagger \otimes m} = \sum_{\sigma, \pi \in \mathfrak{S}_m} W(\sigma \pi) \Tr (B \psi^d_m(\sigma)) \psi^d_m(\pi),
    \label{eq:weingarten}
\end{equation}
Above, the Weingarten function is defined, for any $\tau \in \mathfrak{S}_m$, as
\begin{equation}
    W(\tau) = \sum_{\lambda \vdash m \atop
    \operatorname{len}(\lambda) \leq d} \frac{m_\lambda^2}{(m!)^2 d_\lambda} \chi_\lambda(\tau),
\end{equation}
where $d_\lambda$ and $m_\lambda$ are the dimension and multiplicity of the $\mathbf{U}_d$ irrep $\lambda$ in Schur-Weyl duality, and $\chi_\gamma(\tau) = \Tr \mathbf{p}_\lambda(\tau)$ is the character of the $\mathfrak{S}_m$ irrep $\lambda$, cf. Theorem~\ref{thm:schur-weyl}.

For $m=2$ we have two irreps $\yd[1]{2}$ and $\yd[1]{1,1}$, with $m_{\yd[0.5]{2}} = m_{\yd[0.5]{1,1}}=1$, $d_{\yd[0.5]{2}} = \frac{d(d+1)}{2}$, $d_{\yd[0.5]{1,1}}= \frac{d(d-1)}{2}$, $\chi_{\yd[0.5]{2}}(e)=\chi_{\yd[0.5]{2}}((1\ 2)) = 1$, $\chi_{\yd[0.5]{1,1}}(e)=1, \chi_{\yd[0.5]{1,1}}((1\ 2)) = -1$. Thus, according to Eq.~\eqref{eq:weingarten} we have
\begin{equation}
    \int_{\mathbf{U}_d} d\mu(V) V^{\otimes 2} B V^{\dagger \otimes 2}=\frac{1}{d^2-1}\left(\Tr[B]-\frac{\Tr[B\,\, \mathrm{SWAP}]}{d}\right)\id\otimes \id +\frac{1}{d^2-1}\left(\Tr[B\,\, \mathrm{SWAP}]-\frac{\Tr[B]}{d}\right) \mathrm{SWAP},
\end{equation}
and therefore we find $M$ to be
\begin{equation}
    M = \frac{C}{d^2(d^2-1)}\left( \id_d^{\otimes 4} + \mathrm{SWAP} \otimes \mathrm{SWAP} - \frac{\mathrm{SWAP} \otimes \id_{d}^{\otimes 2} +  \id_{d}^{\otimes 2} \otimes \mathrm{SWAP} }{d} \right).
\end{equation}
Choosing $C=\frac{d^3(d-1)}{2}$ then ensures $\|M\| \leq 1$, giving rise to a success probability of $C/d^4 = (d-1)/2d$.
Thus, the probability of success is at least $1/4$ and asymptotes $1/2$ when $d$ is large.

Cases of larger $m$ can be similarly analyzed. There is unlikely a closed-form formula for a tight bound on the success probability for general $m$. However, we can readily see that the success probability is a constant when $m=O(1)$ and $d \gg m$ as follows. Observe that $m_\lambda$, $(m!)^2$, $\chi_\lambda(\tau)$, and $|\{\lambda: \lambda \vdash m\}|$ are all $O(1)$. In addition, we can see from the Weyl dimension formula in Section~\ref{sec:GTbasis} that $d_\lambda = \Omega(d^m)$. Thus, $W(\tau) = O(\frac{1}{d^m})$ for any $\tau \in \mathfrak{S}_m$. We can also bound the operator norm of each summand in Eq.~\eqref{eq:weingarten} when applied to $B = \ket{\Phi_{d}}\bra{\Phi_d}^{\otimes m}$:
\begin{equation}
    \| \Tr_A (\ket{\Phi_{d}}\bra{\Phi_d}^{\otimes m} (\id_{d^m} \otimes \psi^d_m(\sigma))) (\id_{d^m} \otimes \psi^d_m(\pi))\| \leq  \frac{1}{d^m}\| \id_{d^m} \otimes \psi^d_m(\pi)\| \leq \frac{1}{d^m},
\end{equation}
where $\Tr_A$ is performed on Alice's half of the Choi state, and the first inequality follows from the equality $\Tr (\ket{\Phi_{d}}\bra{\Phi_d}^{\otimes m} (Q \otimes Q')) = \frac{1}{d^m} \Tr(Q^\top Q')$ and Holder's inequality. Therefore, we can bound $\|M\|$ by $O(\frac{C}{d^{2m}})$ (again, we are treating $m=O(1) \ll d$). So it suffices to choose $C= \Omega(d^{2m})$ to guarantee $\|M\| \leq 1$. This leads to $\Prob[M] = \Omega(1)$.

\subsubsection{Estimating Pauli-sparse observables with shadow tomography}
In many practical cases, one needs not obtain the full output state. Rather, the ability to estimate expectation values of observables on the output state may already suffice for one's interest. For example, if the output $n$ is small (e.g., $n=1$ in quantum majority vote~\cite{buhrman2022quantum}, universal spin flip machine~\cite{gisin1999spin}, or more generally any basis-independent evaluations of Boolean functions), then the ability to estimate all Pauli observables essentially implies the ability to reconstruct the output in time $\poly(d)$. Consider the following task: given $M$ observables $O_1,\hdots,O_M \in \mcB((\mathbb{C}^2)^{\otimes n})$ and an input state $\rho \in \mcB((\mathbb{C}^2)^{\otimes m})$, estimate each of $\Tr[\mcN(\rho) O_i]$ to precision $\varepsilon$.

We make the observation that this is exactly what shadow tomography can achieve in several settings. Below, we outline two such protocols that are applicable when one can prepare the Choi state of the channel. We consider the same setting as in Assumption~\ref{assumption}, which allows us to apply the mixed Schur transform to prepare the Choi state $J^\mcN$. After this step, we can apply the results below due to~\cite{caro2022learning} and~\cite{huang2022learning} for qubit channels. All of them should morally generalize to qudits by replacing $n$ and $m$ by $n \log_2 d$ and $m\log_2 d$, respectively. We note that these results apply to general Choi states, thus it seems likely that they can be significantly improved for channels with symmetry as considered in this paper\footnote{For example, one can keep the Choi state as it is in the Schur basis and instead apply the mixed Schur transform to the states/observables, then apply shadow tomography within each diagonal irrep block.}. However, we will not attempt to do this here. 

The first result is due to Caro~\cite{caro2022learning}, which applies under certain sparsity assumptions. 
\begin{definition}[Pauli sparse states and  observables] An observable $O \in (\mathbb{C}^2)^{\otimes n}$ is $s_O$-sparse if its Pauli expansion has at most $s_O$ nonzero coefficients. An $s_\rho$-sparse input quantum state is similarly defined.
\end{definition}

\begin{prop}[Predicting Pauli-sparse expectation values for arbitrary channels, Corollary 4.2 in~\cite{caro2022learning}] There is a procedure that uses $\mcO(\frac{m+\log(1/\delta)}{\varepsilon^4} B^4 s^4_\rho s_O^2)$ copies of the Choi state $J^\mcN$ to produce a classical description $\Tilde{\mcN}$ of $\mcN$. From $\Tilde{\mcN}$, any $M$ expectation values of the form $\Tr(O_i \mcN(\rho)), \leq i \in [M]$ where $O_i$ are $s_O$-Pauli-sparse observables satisfying $\|O_i\|_2 \leq B 2^{n/2}$ and $\rho$ is $s_\rho$-Pauli-sparse state can be predicted to accuracy $\varepsilon$ with probability $1-\delta$ using $\mcO(\frac{\log (\min\{M s_\rho s_O, 16^m\})+\log(1/\delta))}{\varepsilon^2}B^2 s_\rho^2 s_O )$ additional copies of $J^\mcN$ and classical computation time and space $\mcO(\frac{m^2+m\log (1/\delta)}{\varepsilon^4} B^4 s_\rho^4 s_O^2M)$.
\label{prop:caro}
\end{prop}

Their result can naturally be generalized to approximately Pauli-sparse observables and states or under sparsity assumptions in other reasonable operator bases, see~\cite{caro2022learning}. Important examples of Pauli-sparse observables include Pauli operators of arbitrary weight, arbitrary local observables (see Examples 3.4 in~\cite{caro2022learning}), and the observables in Proposition~\ref{prop:huang} below. Pauli-sparse input states, however, are somewhat more restricted. Examples include  any $m=O(1)$-qubit states  and stabilizer subspace states, $\frac{1}{2^m}\prod_{\ell=1}^{k} (\frac{\id_{2^m}+ P_\ell}{2})$, which are $2^k$-Pauli-sparse. According to Remark 3.7 in~\cite{caro2022learning}, any Pauli-sparse states necessarily have rank $\Omega(2^m/\poly(m))$.

The above protocol works for any states and observables under sparsity assumptions. The next result due to Huang, Chen, and Preskill~\cite{huang2022learning}, on the other hand, will work for any sparse observables but only guarantees good average-case precision over certain distributions over the input state. We observe that their algorithm can be adapted to the case when only the Choi state of the channel is provided.

\begin{prop}[Predicting arbitrary observables on average, adaptation of case $\epsilon'=0$ in Theorem 14 in~\cite{huang2022learning}] Consider any quantum state distribution $\mcD$ that is invariant under single-qubit Clifford gates and any observable $O$ that is a linear combination of few-body Hermitian operators such that each qubit participates in $O(1)$ of them, with $\|O\|\leq 1$. Given a constant $\varepsilon$, there is a procedure that uses $N=\log (\frac{m}{\delta}) 2^{\mcO(\log(1/\varepsilon) \log m)}$ copies of the Choi state $J^\mcN$ and additionally $\mcO(mN)$ classical computation time to output an estimate $e(\rho,O)$ such that the following holds with probability $\geq 1-\delta$
\begin{equation}
    \mathbb{E}_{\rho \sim \mcD} |e(\rho,O) - \Tr (O \mcN(\rho)) |^2 \leq \varepsilon .
\end{equation}
\label{prop:huang}
\end{prop}

\begin{proof}
    The main step in their algorithm is obtaining a classical shadow of the channel and Heisenberg-evolved operator $\mcN^\dagger(O)$ (Equation 7 in~\cite{huang2022learning}). This was performed by (1) sampling a random input product state, $\ket{\psi}=\bigotimes_{i=1}^{m} \ket{s_i}$, with $\ket{s_i} \in \{\ket{0},\ket{1},\ket{+},\ket{-}, \ket{+i},\ket{-i}\}$ are the eigenvectors of Pauli $Z,X,Y$ operators, (2) sending $\ket{\psi}$ through the channel, and (3) performing a random Pauli measurement on each output qubit. Note that steps (1) and (2) can as well be achieved using the Choi state by measuring in a random Pauli basis the first register of the Choi state $J^\mcN = \frac{1}{2^m} \sum_{i,j \in [2^m]} \ket{i}\bra{j} \otimes \mcN(\ket{i}\bra{j})$. The resulting state in the second register will then be $\mcN(\ket{\psi}\bra{\psi})$ for a random product stabilizer state $\ket{\psi}$ and we know what $\ket{\psi}$ is from the measurement outcomes on the first register.
\end{proof}

Finally, we note that the above results were proved and stated assuming the input and output Hilbert spaces have the same dimension $n=m$ in the respective references. Given the mild dependence on $n$ in Proposition~\ref{prop:caro}, we can simply pad, either the input or output, with a dummy fixed state to apply it to the general case when $n \neq m$. On the other hand, upon inspecting the proof of Proposition~\ref{prop:huang} in~\cite[Appendix D.4.d]{huang2022learning}, it can be seen that the exponent in the major term of the sample complexity $N$ depends on the input dimension $m$ but not $n$. Intuitively, this is because their algorithm is based on learning the Heisenberg-evolved observable $\mcN^\dagger(O)$, which lives in the input Hilbert space.

The conditions on the input state or its distribution in the above protocols might make them applicable in certain applications but not others. For example, Proposition~\ref{prop:huang} may be applicable to multiport-based teleportation~\cite{ishizaka2008asymptotic, mozrzymas2021optimal}, where the input state can be arbitrary (think $\mcD$ is the Haar-random distribution), while it is unlikely so for quantum majority vote, where the input state is a tensor product of the same qudit state (and thus the distribution is not invariant under local gates). However, the enormous practical advantage of these protocols is that they use very little quantum processing resources: Proposition~\ref{prop:caro} only uses 2-copy Bell measurements and Proposition~\ref{prop:huang} only performs qubit-wise operations on the Choi state.

\subsubsection{Block-encoding-based implementation}
Next, we consider an alternative setting based on the block-encoding framework~\cite{low2019hamiltonian, gilyen2019quantum}. The limitations will depend on the spectral properties of the Choi state. The method allows efficiently estimating the expectation value $\Tr(O \mcN(\rho))$ of an observable $O$ under the following assumption.

\begin{assumption}
One can efficiently prepare a quantum state $\sigma \in \mcB((\mathbb{C}^2)^{\otimes n})$ such that $\sigma = \frac{O^\top + c\id}{Z}$ for some normalization coefficients $c, Z$.
\label{a2}
\end{assumption}

For example, this is the case when $n=O(1)$ (so that $\mcN(\rho)$ can also be fully reconstructed in $\poly(n,d)$ time), when $O$ is low-rank, or when $O$ itself is an $n$-qudit quantum state.

The main idea is to use the following characterization of the \textit{unnormalized} Choi state
\begin{equation}
    \mcN(\rho) = \Tr_A [\Tilde{J}^\mcN_{AB} ( \rho_A^\top  \otimes \id_B)],
    \label{eq:choi-channel}
\end{equation}

We will also require that a block-encoding of the Choi state can be efficiently prepared. A standard situation where this is the case is when the Choi state is sparse in the Schur basis
and its entries are efficiently computable. This includes the cases considered in~\cite{grinko2022linear}, where the Choi state is diagonal in the Schur basis. However, there are efficient block encodings for certain structured classes of dense matrices as well~\cite{sunderhauf2023block, nguyen2022block}, 
so the methods presented here may apply more generally. We note that block-encodings can also be used to prepare quantum states, as required in Assumption~\ref{assumption}.

\begin{definition}[Block encoding \cite{gilyen2019quantum}] Suppose that $A$ is an $n$-qubit operator, $\alpha, \varepsilon>0$, then we say that an $(n+a)$-qubit unitary $U$ is an $(\alpha, \varepsilon)$-block-encoding of $A$ if
\begin{equation*}
\left\|A-\alpha(\bra{0}^{\otimes a} \otimes I) U (\ket{0}^{\otimes a} \otimes I )\right\| \leq \varepsilon,
\end{equation*}
where $\|\cdot\|$ denotes the operator norm of a matrix.
\end{definition}
Intuitively, this means $A$ is ``stored'' in the top left block of the unitary $U$
\begin{equation*}
    U \approx \begin{pmatrix}
        A/\alpha & * \\
        * & * 
    \end{pmatrix}.
\end{equation*}
Here, $\alpha \geq \|A\|$ for $U$ to be a unitary.
As an example, a unitary is $(1,0)$-block-encoding of itself. Note that often in the literature~\cite{gilyen2019quantum} the number of ancilla qubits $a$ is included in the block-encoding notation. We omit it for brevity here since we will only ever use $\Tilde{\mcO}(n)$ ancillas in this work.

Block encodings can be used as follows. Given an $n$-qubit input state $\ket{\psi}$, applying the unitary $U$ to
the state $\ket{0^a}\ket{\psi}$, measuring the first register, and postselecting on the $\ket{0^a}$ outcome, the second system contains a state proportional $A \ket{\psi}$.

Next, we describe the standard input model in the block-encoding framework.

\begin{definition}[Oracle access model]\label{def:sparse-oracle}
Let $A \in \mathbb{C}^{2^n \times 2^n}$ be a $t$-sparse matrix with bounded entries $\max_{i,j} |a_{ij}| \leq 1$. Suppose that we have the access to the following $2(n+1)$-qubit oracles
\begin{equation*}
\begin{aligned}
    &O_r: \ket{i}\ket{k} \rightarrow \ket{i} \ket{r_{ik}} & 
    0\leq  i < 2^n, 0 \leq k < t,\\
    &O_c: \ket{l}\ket{j} \rightarrow \ket{c_{lj}} \ket{j} & 0 \leq j < 2^n, 0 \leq l <t,
\end{aligned}
\end{equation*}
where $r_{ik}$ is the index for the $k$-th non-zero entry of the $i$-th row of A, or if there are less than $k$ non-zero
entries, then $r_{ik}=k+2^n$, and similarly $c_{lj}$ is the index for the $l$-th non-zero entry of the $j$-th column of $A$, or if there are less than $l$ non-zero entries, then $c_{lj}=l+2^n$. Additionally, the following oracle is provided
\begin{equation*}
    O_{A}: \ket{i}\ket{j}\ket{0}^{\otimes b} \rightarrow \ket{i}\ket{j}\ket{\Tilde{a}_{ij}}, \hspace{0.5cm}  0 \leq i,j < 2^n,
\end{equation*}
where $\Tilde{a}_{ij}$ is the $b$-qubit description of $a_{ij}$ with $b=\Theta\log(1/\varepsilon)$, and if $i$ or $j$ is out of range then $\Tilde{a}_{ij}=0$.
\end{definition}

\begin{lemma}[Block-encoding of oracle-access matrices, Lemma 48 in~\cite{gilyen2019quantum}] Provided with oracle access to a $t$-sparse matrix $A$ whose entries are bounded $B$, one can implement a $(tB, \varepsilon)$-block-encoding of $A$ with a single use of $O_r, O_c$, two uses of $O_A$, $n+3$ ancillas, and additionally $\mcO(n+\operatorname{polylog}(\frac{tB}{\varepsilon}))$ two-qubit gates.
\label{lem:sparse}
\end{lemma}

We say a matrix is sparse when $t = \poly(n)$, in which case many common matrix operations on $A$ can be performed efficiently~\cite{gilyen2019quantum}. In the setting of unitary-equivariant channels, we require that the entries of the Choi state \textit{in the Schur basis} are an efficiently computable function of the indices in order to use this oracle model. For example, this can be the output of the algorithm in~\cite{grinko2022linear}.

\begin{lemma}[Block-encoding Schur-sparse Choi state and its square root]
\label{lem:blockencodeChoi} Suppose $\mcN$ is an $n$-to-$m$ qudit unitary-equivariant channel such that its unnormalized Choi matrix $\Tilde{J}^\mcN$ is $t$-sparse in the Schur basis. Furthermore, suppose the entries of $\Tilde{J}^\mcN$ in the Schur basis are bounded above by $B$ and the minimum nonzero eigenvalue of $\Tilde{J}^\mcN$ is bounded from below by $\kappa$. Then using the oracle in Definition~\ref{def:sparse-oracle}, one can implement a $(tB, \varepsilon )$-block-encoding of $\Tilde{J}^\mcN$ using $a=\Tilde{\mcO}((n+m)d)$\footnote{The $\widetilde{\mcO}(\cdot)$ notation hides $\log n$, $\log m$, $\log d$, and $\log (1/\varepsilon)$ factors, where $\varepsilon$ is the target precision in diamond norm. 
 } ancillas and $\Tilde{\mcO}((n+m)d^4)$ two-qubit gates. Furthermore, if $\Tilde{J}^\mcN$ is diagonal in Schur basis~\cite{grinko2022linear} ($t=1$), then one can implement a $(\sqrt{B},\varepsilon)$-block-encoding of $(\Tilde{J}^\mcN)^{1/2}$ with the same complexity. Generally if $t > 1$, then one can instead obtain a $(\sqrt{tB},\varepsilon)$-block-encoding with complexity $\Tilde{\mcO}((n+m)d^4tB/\kappa)$.
\end{lemma}
\begin{proof}
    Denote by $\Tilde{J}_\mathrm{Sch} = W_{\mathrm{Sch}(n,m)}^\dagger \Tilde{J}^\mcN W_{\mathrm{Sch}(n,m)}$ the unnormalized Choi matrix in the Schur basis, which is $t$-sparse by assumption.  
    Applying Lemma~\ref{lem:sparse} we get a $(tB, \varepsilon)$-block-encoding $E_\mathrm{Sch}$ of $\Tilde{J}_\mathrm{Sch}^\top$ using $a= (n+m) \log d +3$ ancillas and gate complexity $\Tilde{\mcO}((n+m)d)$. Next we obtain a $(tB, \varepsilon )$-block-encoding $E$ of $(\Tilde{J}^\mcN)\top$ by taking $(\id \otimes W_{\mathrm{Sch}(n,m)}^*) E_\mathrm{Sch} (\id \otimes W_{\mathrm{Sch}(n,m)}^\top)$, where the identity operator $\id$ acts on the ancilla qubits of the block encoding. Note that the transpose and the complex conjugate of $W_{\mathrm{Sch}(n,m)}$ can be implemented by reversing the gate order and taking the complex conjugate of the local gates of the mixed Schur transform, respectively. As shown in the previous sections, the gate complexity of the mixed Schur transform is $\Tilde{\mcO}((n+m)d^4)$.

    Next, we describe how to implement a block-encoding of $(\Tilde{J}^\mcN)^{1/2}$. This can be easily achieved if $\Tilde{J}_\mathrm{Sch}$ is diagonal as considered in~\cite{grinko2022linear}, as we can instead query the square root of the entries of it and get a $(\sqrt{B},\varepsilon)$-block-encoding. If $t\neq 1$, we can instead obtain $(\Tilde{J}^\mcN)^{1/2}$ from the block-encoding of $\Tilde{J}^\mcN$ by using the quantum singular value transform to perform a degree-$\Tilde{\mcO}(tB/\kappa)$ polynomial approximating the function $f(x)=\sqrt{x}$~\cite[Corollary 55]{chakraborty2018power}. 
\end{proof}

\begin{prop}[Block-encoding-based implementation]
\label{prop:amplitude} Suppose we can efficiently prepare a quantum state $\sigma \in \mcB((\mathbb{C}^2)^{\otimes n})$ as described in Assumption~\ref{a2}. In addition, suppose we can prepare a $(\sqrt{tB},\frac{\varepsilon}{tBZ})$-block-encoding $U$ of $(\Tilde{J}^\mcN)^{\top 1/2}$.
Then, we can estimate $\Tr(O \mcN(\rho))$ to precision $2\varepsilon$ using $\mcO(tBZ/\varepsilon)$ queries to $U$ and elementary gates.
\end{prop}
\begin{proof} 
    Let $U$ be the block encoding of $(\Tilde{J}^\mcN)^{\top 1/2}$ and let $a$ be the number of ancilla qubits it uses. We will apply $U$ on $\ket{0^a}\bra{0^a}\otimes \rho \otimes \sigma $ and then measure the ancilla register. Notice that the probability of measuring $\ket{0^a}$ is
    \begin{align*}
        &\Tr ( \bra{0^a} U \ket{0^a} \rho \otimes \sigma \bra{0^a} U^\dagger \ket{0^a}  )\\
        & \overset{\frac{\varepsilon}{tBZ}}{\approx} \frac{1}{tB} \Tr ((\Tilde{J}^\mcN)^{\top 1/2} \rho \otimes 
        \sigma (\Tilde{J}^\mcN)^{\top 1/2} )\\
        & =  \frac{1}{tB Z} \Tr ((\Tilde{J}^\mcN)^{\top} \rho \otimes  (O^\top +c \id) )  \\
        & =  \frac{1}{tB Z} (\Tr (\mcN(\rho)  O ) +c) .
    \end{align*}
Thus, estimating this probability to precision $\varepsilon/tBZ$ using $\mcO(tBZ/\varepsilon)$ iterations of amplitude estimation~\cite{brassard2002quantum} gives a $2\varepsilon$-error estimate of $\Tr (\mcN(\rho) O))$.
\end{proof}
We now combine Lemma~\ref{lem:blockencodeChoi} and Propostion~\ref{prop:amplitude}. First of all, note that Lemma~\ref{lem:blockencodeChoi} can be easily modified to prepare a block encoding of $(\Tilde{J}^\mcN)^{\top 1/2}$ (note the transpose). Therefore, for Choi state diagonal in the Schur basis ($t=1$), we can estimate $\Tr(\mcN(\rho) O)$ to precision $2\varepsilon$ using $\mcO(tBZ)$ queries to the oracles in Assumption~\ref{a2} and Definition~\ref{def:sparse-oracle} and
$\Tilde{\mcO}(tBZ(n+m)d^4)$ gate complexity. For Choi states that are $t$-sparse in Schur basis, the previous complexities become $\mcO(tBZ)$ queries and
$\Tilde{\mcO}(t^2 \frac{B^2}{\kappa} Z(n+m)d^4)$ gate complexity, respectively.

In addition, we note that compositioning channels is straightforward in the block-encoding framework. For example, $\mcN_{A \rightarrow B}$ and $\mcN'_{B\mapsto C}$ can be composed by multiplying the two unitaries block-encoding $(\Tilde{J}^\mcN)^{\top 1/2}_{AB}\otimes \id_C$ and $\id_A \otimes  (\Tilde{J}^{\mcN'})^{\top 1/2}_{BC}$. This operation can be done efficiently by~\cite[Lemma 53]{gilyen2019quantum}.

\begin{remark}
Let us comment on the parameters $B,\kappa$, and $Z$. Here, $Z$ is the normalization coefficient when encoding the observable $O$ into a quantum state, so it is roughly the trace norm $\|O\|_1$ of the observable $O$, which is small when $n=O(1)$, when $O$ is low-rank, or when $O$ itself is a quantum state. These complement the Pauli-sparse observables considered in the shadow tomography protocols in Propositions~\ref{prop:caro} and~\ref{prop:huang}. Next, the parameters $B$ and $\kappa$ depend on the spectral properties of the unnormalized Choi state $\Tilde{J}^\mcN$. If the sparsity $t$ is small, it can be seen that $B$ is roughly the operator norm of the \textit{unnormalized} Choi state. On the other hand, $\|\Tilde{J}^\mcN\|_1 = d^m$. Thus, for $B$ to be small, i.e. $\poly(n,m,d)$, we necessarily require $\Tilde{J}^\mcN$ to have $\Omega(\frac{d^m}{\poly(n,m,d)})$ rank or equivalently, the channel $\mcN$ should have high Kraus-rank. When $t > 1$, the complexity further depends on the ratio $B/\kappa$, which is roughly the condition number of $J^\mcN$. So we further require the condition number of $J^\mcN$ to be small in this case.
\end{remark}

\subsection{Other potential applications}\label{sec:other-applications}
We discuss other possible applications and open questions regarding the mixed Schur transform.

\paragraph{Port-based teleportation} In the previous section, we gave several alternatives to implement \textit{general} unitary-equivariant channels. However, one can likely exploit the finer details of specific applications for this. In a concurrent work~\cite{grinko2023gelfand}, the authors give an efficient implementation for the measurements in (single)port-based teleportation~\cite{ishizaka2008asymptotic} by using the mixed Schur transform with $m=1$ (see also a concurrent work of~\cite{fei2023efficient}). Port-based teleportation is a seminal protocol~\cite{ishizaka2008asymptotic} that, surprisingly, removes the need for the correction step of the receiving party. It has applications in universal programmable quantum processors, quantum cryptography~\cite{beigi2011simplified}, black-box unitary transformation~\cite{yoshida2022reversing}, establishing connections between quantum communication complexity and Bell's inequality violations~\cite{buhrman2016quantum}, etc.
Constructing an efficient implementation for the multi-port case, where Ref.~\cite{mozrzymas2021optimal} has recently settled the solutions for the optimal performance and measurements, is thus of great interest.

\paragraph{Quantum SDP solver with symmetry constraints} Quantum SDP solvers~\cite{brandao2017quantum, van2018improvements} are an important recent advance in quantum algorithms that offer quadratic speedups in certain settings.
They are based on quantizing the seminal algorithm of Arora and Kale~\cite{arora2007combinatorial}. It is evident from the Arora-Kale algorithm (cf.~\cite[Section 2]{van2018improvements}) that the output solution inherits the block structure from the constraint matrices, thus it can be adapted to SDP with symmetry constraints as long as the irrep decomposition of the symmetry is efficiently computable.
For this reason, we believe that applying quantum SDP solvers to problems with symmetry constraints could gain larger speedups being combined with efficient quantum implementations of representation-theoretic transforms--another thing that quantum computers are good at~\cite{moore2006generic}.
The mixed Schur transform studied in this work could be used as a subroutine in such problems with unitary-equivariant or partially transposed permutation constraints.

\paragraph{Learning and property testing} Testing symmetries of quantum systems is a fundamental task in physics and quantum information. Recent algorithms~\cite{laborde2021testing} have been proposed for finite groups (using quantum Fourier transforms) and small instances of continuous symmetries.
The mixed Schur transform should naturally be useful in testing unitary-equivariance of quantum processes, a  prototypical case of which is unitary dynamics with $\mathbf{U}_d$-invariant Hamiltonian. Mixed tensor representations also appear in the notion of strong unitary-designs~\cite{nakata2021quantum}, where they are applied to high-order randomized benchmarking protocols and testing self-adjointness of quantum channels. Another important task is testing time-reversal symmetry of a unitary (i.e. whether $U=\Bar{U}$). One natural approach is perhaps to pretend $U^{\otimes (n+m)}=U^{\otimes n} \otimes \Bar{U}^{\otimes m}$ and perform both Schur and mixed Schur sampling~\cite{childs2007weak} and compare their outcomes.

In unitary process learning, one can ask whether being given access to the $\Bar{U}$ in addition to $U$ can improve the learner's sample or time complexities. Vice versa, black-box unitary transformation protocols such as $U \mapsto U^\dagger$~\cite{yoshida2022reversing} and $U \mapsto \Bar{U}$~\cite{grinko2022linear} may have implications to whether having additional accesses to $U^\dagger$ or $\Bar{U}$ helps in unitary process tomography.

\begin{remark} The mixed Schur transform is unlikely to be applicable to quantum states. This is because the dual irrep $\Bar{U}$ should generally be $(A^{-1})^\top$ for general invertible matrices $A \in \mathbf{GL}_n$. However, this would translate to $(\rho^{-1})^\top/\Tr(\rho^{-1})$ for quantum states.
\end{remark}

\paragraph{Extension of permutational quantum computing} Estimating irrep matrix entries of $\mathfrak{S}_n$ (i.e. strong Schur sampling) is a defining problem of the complexity class Permutational Quantum Polynomial (PQP) introduced by Jordan~\cite{jordan2010permutational}. Roughly speaking, the problem is to estimate the amplitude $|\bra{\lambda,q_\lambda,p_\lambda}_\mathrm{Sch} \psi_n^d (\pi) \ket{\mu,q_\mu,p_\mu}_\mathrm{Sch} |^2$ to additive error for given Schur basis vectors $\ket{\lambda,q_\lambda,p_\lambda}_\mathrm{Sch}$, $\ket{\mu,q_\mu,p_\mu}_\mathrm{Sch}$, and a permutation $\pi$ (cf. Eq.~\eqref{eq:psi_n^d}). This task can be achieved using the Schur transform.  However, it was shown to be classically tractable~\cite{havlivcek2018quantum}. To overcome this,~\cite{zheng2022super} extended PQP via Hamiltonian simulation techniques. Their idea is to replace the permutation $\psi_n^d (\pi)$ by the unitary evolution of a Hamiltonian $H$ that is a linear combination of polynomially many different permutations $\psi_n^d (\pi)$. Similarly, the mixed Schur transform can be used to strongly sample from $\mathcal{A}^d_{n,m}$ irreps, thus offering yet another extension of PQP (one may call it ``partially transposed'' PQP).
In particular, we can estimate (and sample from) mixed Schur basis elements $|\bra{\gamma,q_\gamma,p_\gamma}_\mathrm{Sch} e^{-iHt} \ket{\mu,q_\mu,p_\mu}_\mathrm{Sch} |^2$, where the Hamiltonian $H$ is a linear combination of walled Brauer diagrams. For example, $H$ could be an antiferromagnetic Lie-superalgebra spin chain studied in~\cite{candu2011continuum}. The basis elements $\ket{\gamma,q_\gamma,p_\gamma}_\mathrm{Sch}$ can be prepared efficiently using the mixed Schur transform, whereas $H$, being a polynomially sparse Hamiltonian, can be efficiently simulated by standard block-encoding techniques~\cite{low2019hamiltonian, gilyen2019quantum}. Therefore, this problem can be solved efficiently on a quantum computer, and generalizes both PQP~\cite{jordan2010permutational} and PQP+~\cite{zheng2022super}. On the other hand, for the same reasoning as in~\cite{zheng2022super} - namely the symmetric group Fourier transform takes classical exponential time, we conjecture that this model is not classically simulable.

\paragraph{Quantum complexity of representation-theoretic problems} Recent works~\cite{ikenmeyer2023remark, bravyi2023quantum} studied the quantum complexity of Kronecker coefficients. Roughly speaking, the Kronecker coefficients are the analog of Clebsch-Gordan coefficients for $\mathfrak{S}_n$. While it is still unknown if computing them is in $\mathsf{\#P}$, the authors showed that this problem is in $\mathsf{\#BQP}$ using the generalized phase estimation algorithm~\cite[Section 8]{harrow2005applications}, which is in turn based on the $\mathfrak{S}_n$ quantum Fourier transform~\cite{beals1997quantum}.
What is the complexity of the (generalized) Kronecker coefficients of the walled Brauer algebra? While it seems reasonable they should have the same complexity as the regular Kronecker coefficients, a prerequisite would be an efficient quantum circuit for a Fourier transform of the walled Brauer algebra. It is thus an interesting open question if such a transform exists, and if so whether it could be achieved using the mixed Schur transform.

\printbibliography

\appendix

\end{document}